\documentclass[11pt, a4paper]{article}

\usepackage[T1]{fontenc}
\usepackage[utf8]{inputenc}
\usepackage{lmodern} % Simulates the 'lmodern' used in EJS/Bernoulli
\RequirePackage[numbers]{natbib}
\RequirePackage{graphicx}
\usepackage{algorithm}
\usepackage[noend]{algpseudocode}
\usepackage{bbm}
\usepackage{booktabs}
\usepackage{verbatim}
\usepackage{tabularx}
\usepackage{array} 
\usepackage{amsmath}
\usepackage{amsthm}
\usepackage{amssymb}
\usepackage{bm}      % For bold math symbols

\usepackage{graphicx}
\usepackage{booktabs} % For professional tables
\usepackage{enumitem} % For customizing lists
\usepackage[margin=1in]{geometry} % Sets standard 1-inch margins

\usepackage[colorlinks=true, linkcolor=blue, citecolor=blue, urlcolor=blue]{hyperref}

\newcommand\numberthis{\addtocounter{equation}{1}\tag{\theequation}}

\theoremstyle{plain}
\newtheorem{theorem}{Theorem}[section]
\newtheorem{lemma}[theorem]{Lemma}

\theoremstyle{definition}
\newtheorem{definition}{Definition}[section]

\theoremstyle{remark}

\title{\textbf{Hybrid Partial Least Squares Regression \\
with Multiple Functional and Scalar Predictors}}

\author{
  \textbf{Jongmin Mun} \\
  Department of Data Sciences and Operations, Marshall School of Business \\
  University of Southern California \\
  \texttt{jongmin.mun@marshall.usc.edu}
  \and
  \textbf{Jeong Hoon Jang} \\
  Department of Biostatistics and Data Science \\
  University of Texas Medical Branch \\
  \texttt{jejang@utmb.edu}
}

\date{\today}

\begin{document}

\maketitle

\begin{abstract}
Motivated by renal imaging studies that combine renogram curves with pharmacokinetic and demographic covariates, we propose  Hybrid partial least squares (Hybrid PLS)   for simultaneous supervised dimension reduction and regression in the presence of cross-modality correlations.
The proposed approach embeds   multiple functional and scalar predictors into a unified hybrid Hilbert space and  rigorously extends the nonlinear iterative PLS (NIPALS) algorithm. This theoretical development is complemented by a sample-level algorithm that incorporates roughness penalties to control smoothness. By exploiting the rank-one structure of the resulting optimization problem, the algorithm admits a computationally efficient closed-form solution that requires solving only linear systems at each iteration.
We establish fundamental geometric properties of the proposed framework, including orthogonality of the latent scores and PLS directions. Extensive numerical studies on synthetic data, together with an application to a renal imaging study, validate these theoretical results and demonstrate the method's ability to recover predictive structure under intermodal multicollinearity, yielding parsimonious low-dimensional representations.
    
\end{abstract}
	\section{Introduction} \label{sec: Introduction}
In biomedical studies, responses are often influenced by both scalar covariates and multiple functional trajectories.
Such settings motivate the study of partially functional linear models (PFLM; \citealp{Ramsay2005}) with $K$ functional predictors and $p$ scalar predictors \citep{xu_estimation_2020}.
For notational simplicity, we assume throughout that the response and predictors are centered.
Let us denote  $\langle f, g \rangle_{\mathbb{L}_2} := \int f(t)g(t)\,dt$.
The PFLM links the response $Y_i \in \mathbb{R}$ to the predictors via
\begin{equation*}\label{PFLM}
Y_i = \mathbf{Z}_i^\top \boldsymbol{\beta} + \sum_{k=1}^K \langle X_{ik}, \beta_k \rangle_{\mathbb{L}_2} + \epsilon_i, \quad i = 1, \ldots, n,
\end{equation*}
 where $\mathbf{Z}_i \in \mathbb{R}^p$ denotes scalar covariates with regression coefficients $\boldsymbol{\beta}$, $X_{ik}$ is the $k$-th functional predictor with coefficient function $\beta_k$, and $\epsilon_i$ is an additive random error.
\begin{figure}[t!]
    \centering
    \includegraphics[width=0.99\linewidth]{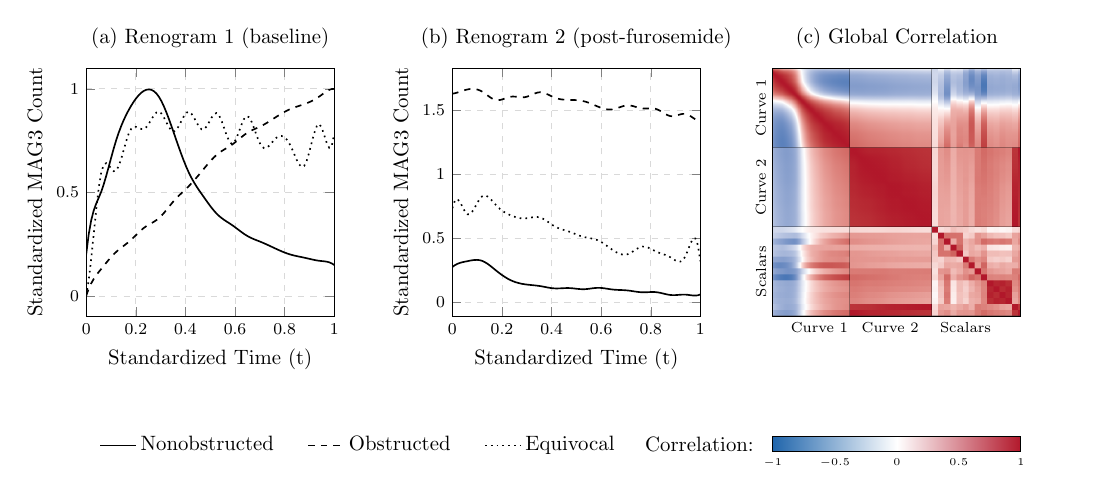}
   \caption{Typical renogram profiles from the Emory University renal study. For each kidney, two functional curves are recorded: (a) the baseline renogram curve and (b) the post-furosemide renogram curve. Lines represent three example cases: nonobstructed (solid), obstructed (dashed), and equivocal (dotted). (c) Heatmap of the correlation matrix for all predictors across 216 subjects, highlighting the strong dependence between the two functional modalities and the fourteen scalar covariates.}

    \label{fig:EDA}
\end{figure}
Analyzing this hybrid model poses two fundamental challenges: achieving parsimony and mitigating multicollinearity.
Parsimonious representations are essential for interpretability, particularly when clinical scores and decision rules are of interest.
Multicollinearity is especially challenging in this setting, as it arises at multiple levels:
(i) within-function correlations, (ii) between-function correlations, (iii) correlations among scalar predictors,
 and (iv) intermodal correlations between functional and scalar predictors.
Ignoring these dependencies can substantially inflate estimation variance, leading to overfitting and 
 degraded predictive performance at deployment.

Our motivating example is the Emory University renal study \citep{changBayesianLatentClass2020, jangPrincipalComponentAnalysis2021}, which collects two renogram curves (functional data) per kidney for the purpose of detecting kidney obstruction. Following intravenous injection of the tracer $^{99\mathrm{m}}$Tc-MAG3, the kidneys filter the tracer and transport it toward the bladder; renogram curves are obtained by tracking the tracer photon (MAG3) count within the kidneys over time. Owing to this tracer-based recording process, previous analyses have typically relied on scalar pharmacokinetic variables---summary measures characterizing tracer uptake, transport, and clearance, such as maximum intensity, time to minimum velocity, and area under the curve---to train predictive models. However, this strategy risks discarding subtle structural information contained in the full functional data, as illustrated in Figures~\ref{fig:EDA}-(a) and \ref{fig:EDA}-(b). A more comprehensive approach is to jointly incorporate the complete functional curves and their scalar summaries within a PFLM, thereby fully leveraging their predictive potential. This joint modeling induces strong correlations between the two modalities, as shown in Figure~\ref{fig:EDA}-(c), making this dataset an ideal test case for the setting considered here. Additional details on the dataset are provided in Section~\ref{section:real_data}.

Most existing methods are designed for homogeneous data structures.
High-dimensional Euclidean regression typically relies on structured regularization, such as lasso \citep{tibshiraniRegressionShrinkageSelection1996},
SCAD \citep{fanVariableSelectionNonconcave2001}, or 
elastic net \citep{zouRegularizationVariableSelection2005}.
Another class of methods relies on dimension reduction techniques. Principal component regression (PCR; \citealp{agarwalRobustnessPrincipalComponent2021}) performs unsupervised dimension reduction. However, in many applications, principal components explaining only a small fraction of the total predictor variance can nonetheless be critical for capturing association with the response. Such cases arise in chemical engineering \citep{smith_critique_1980}, meteorology \citep{kung_regression_1980}, and economics \citep{hill_component_1977}; see \cite{jolliffeNoteUsePrincipal1982} for a detailed discussion. To address this limitation, partial least squares (PLS) regression \citep{woldPathModelsLatent1975} employs supervised dimension reduction.
Functional regression methods, by contrast, focus on processing functional predictors using basis expansions with roughness penalties, functional PCA \citep{reissFunctionalPrincipalComponent2007, aguileraUsingBasisExpansions2010, aguileraPenalizedVersionsFunctional2016}, or functional PLS \citep{predaPLSRegressionStochastic2005}.

PFLM
\citep{shin_prediction_2012, kong_partially_2016, lv_kernelbased_2023}
 integrate functional and scalar predictors through a two-step procedure: functional predictors are first converted into finite-dimensional Euclidean representations, followed by an integrated regression with scalar covariates. However, these approaches overlook the complex joint dependencies between predictors, and the functional component is processed in a response-agnostic manner, as the primary emphasis is on regression coefficient estimation rather than prediction.

To address these limitations, we propose a \textit{hybrid PLS regression} framework. By constructing a unified Hilbert space that jointly represents functional and Euclidean predictors, our approach enables a hybrid extension of the nonlinear iterative partial least squares (NIPALS; \citealp{woldPathModelsLatent1975}) algorithm. The proposed method alternates between (i) extracting joint latent directions that maximize covariance between the predictors and the response and (ii) residualizing both predictors and response. This framework naturally captures all four layers of correlation, preserves predictive power, and yields parsimonious representations. Moreover, by employing a rich basis expansion together with smoothness regularization of the hybrid PLS directions, the method borrows strength across predictors and accommodates dense or irregular functional data. The algorithm remains computationally efficient, relying solely on linear system solutions. The resulting hybrid PLS scores are mutually uncorrelated and provide a clinically interpretable representation of the data.

\subsection{Related literature}
The challenges of achieving parsimony and mitigating multicollinearity have been extensively studied in high-dimensional regression with Euclidean predictors. We briefly review the main methodological developments relevant to our work.

A prominent class of approaches relies on structured regularization, which assumes that only a small, structured subset of predictors drives the signal. Representative methods include the lasso \citep{tibshiraniRegressionShrinkageSelection1996}, SCAD \citep{fanVariableSelectionNonconcave2001}, elastic net \citep{zouRegularizationVariableSelection2005}, adaptive lasso \citep{zouAdaptiveLassoIts2006}, group lasso \citep{yuanModelSelectionEstimation2006}, the Dantzig selector \citep{candesDantzigSelectorStatistical2007}, and sure independence screening \citep{fanSureIndependenceScreening2008}. The fused lasso \citep{tibshiraniSparsitySmoothnessFused2005, tibshiraniSolutionPathGeneralized2011} is particularly relevant for functional settings, as it enforces piecewise-constant structure and smoothness in regression coefficients.

Another major line of research focuses on dimension reduction. Principal component analysis (PCA) regression \citep{agarwalRobustnessPrincipalComponent2021} performs unsupervised reduction based solely on predictor variation, whereas partial least squares (PLS) regression \citep{woldPathModelsLatent1975, dejongSIMPLSAlternativeApproach1993, hellandComparisonPredictionMethods1994} carries out supervised reduction by explicitly incorporating the response. PLS has been widely adopted in high-dimensional genomics \citep{boulesteixPartialLeastSquares2007, chunExpressionQuantitativeTrait2009, chunSparsePartialLeast2010} and  chemometrics  \citep{woldPLSregressionBasicTool2001, broPLSWorks2009, mehmoodReviewVariableSelection2012, mehmoodDiversityApplicationsPartial2016}. Its success has further motivated extensions to multi-way and tensor data, including applications to electrocorticogram-derived movement trajectories \citep{zhaoHigherOrderPartial2013} and single-cell 3D genome organization \citep{parkSparseHigherOrder2024}.

In functional data settings, high dimensionality naturally arises from dense measurements over time or space, making regression problems ill-posed. This is typically addressed by leveraging the smoothness of functional predictors. Standard scalar-on-function linear regression methods combine basis expansions with roughness penalties \citep{cardotSplineEstimatorsFunctional2003, aguileraPenalizedVersionsFunctional2016, caiPredictionFunctionalLinear2006, zhaoWaveletbasedLASSOFunctional2012}.
A related problem, functional data denoising, has been addressed through trend filtering, which imposes piecewise polynomial structure via $\ell_1$ regularization on successive differences \citep{kimell_1TrendFiltering2009, tibshiraniAdaptivePiecewisePolynomial2014, sadhanalaAdditiveModelsTrend2019, wakayamaTrendFilteringFunctional2023}.

Data-driven basis approaches introduce parsimony by representing functional predictors using functional principal components (FPC) or functional partial least squares (FPLS) directions. FPC regression \citep{hallMethodologyConvergenceRates2007, reissFunctionalPrincipalComponent2007, febrero-bandeFunctionalPrincipalComponent2017} provides a low-dimensional representation but may fail to capture the true regression relationship, as components are not guided by the response. In contrast, FPLS \citep{predaPLSRegressionStochastic2005, reissFunctionalPrincipalComponent2007, aguileraUsingBasisExpansions2010, aguileraPenalizedVersionsFunctional2016} constructs orthogonal latent components that explicitly maximize covariance with the response, improving predictive performance. Both theoretical \citep{delaigleMethodologyTheoryPartial2012} and computational \citep{saricamPartialLeastsquaresEstimation2022} developments have further advanced FPLS; for a comprehensive review, see \cite{febrero-bandeFunctionalPrincipalComponent2017}.
To account for correlations among multiple functional predictors, methods such as multiple FPC regression \citep{happ_multivariate_2018} and multiple FPLS \citep{beyaztasRobustFunctionalPartial2022} treat the set of functional predictors as a unified object and derive joint components. However, these approaches do not accommodate mixed predictor types that include both functional and scalar variables.

  Integration of functional and scalar predictors has primarily been studied in the PFLM framework \citep{Ramsay2005}, which focuses on estimating regression coefficients rather than predictive accuracy. Most PFLM approaches follow a two-stage procedure: (i) an unsupervised stage in which each functional predictor is represented in a finite-dimensional basis (e.g., splines or functional principal components) without reference to the response, followed by (ii) (penalized) maximum likelihood estimation combining the vectorized functional predictors with scalar covariates (Table \ref{tab:pflm_summary}).

These two-stage methods have two key limitations: they ignore correlations between functional and scalar predictors, and the functional representations are response-agnostic. While cross-modality correlations have been explored in unsupervised settings such as graphical models \citep{kolar_graph_2014} and precision matrix estimation \citep{gengJointNonparametricPrecision2020}, these approaches do not explicitly handle settings that involve both functional and scalar predictors.

% Define a ragged-right version of X column
\newcolumntype{Y}{>{\raggedright\arraybackslash}X}

\begin{table}[b!]
\centering
\renewcommand{\arraystretch}{1.2} % adds vertical spacing for readability
\caption{Summary of two-stage methods for integrating functional and scalar predictors in PFLM framework}
\vspace{1em}
\label{tab:pflm_summary}
\begin{tabularx}{\textwidth}{l c Y Y} % l = left, c = center, Y = ragged-right auto width
\hline
Reference & Functional predictor & Stage I (Functional processing) & Stage II (Regression) \\
\hline
\cite{zhang_twostage_2007} 
& Single  
& Smoothing splines 
& Linear mixed-effects model \\

\cite{shin_partial_2009} 
& Single  
& Functional PCA 
& Ordinary least squares \\

\cite{shin_prediction_2012} 
& Single  
& Functional PCA 
& Tikhonov regularization \\

\cite{lv_kernelbased_2023}
& Single
& RKHS kernel-based basis (via representer theorem) 
& $\ell_2$ regularization (functional) \& $\ell_1$ regularization (scalar) \\

\cite{kong_partially_2016} 
& Multiple  
& Functional PCA 
& $\ell_1$ regularization \\

\cite{cai_highdimensional_2025} 
& Multiple  
& Functional PCA 
& Weighted least squares \\

\cite{goldsmithPenalizedFunctionalRegression2011} 
& Multiple 
& Functional PCA 
& Regularized linear mixed-effects model \\
\hline
\end{tabularx}
\end{table}

\subsection{Notations}
We use boldface to denote Euclidean vectors. Matrices are represented by bold uppercase Latin or Greek letters, or by blackboard bold characters (e.g., $\mathbb{B}$). In contrast, objects residing in non-Euclidean Hilbert spaces are denoted using non-bold characters.

\subsection{Outline of the Paper}
The remainder of this paper is organized as follows. Section~\ref{section:scalar_pls} reviews Euclidean PLS. Section~\ref{section:hybrid_hilbert} introduces the hybrid Hilbert space framework and the population-level algorithm, while Section~\ref{section:main:our_algorithm} details the sample-level implementation and computational scheme. The theoretical properties of the algorithm are established in Section~\ref{section:sub:geom}. In Section~\ref{section:experiments}, we present simulation studies and a real data analysis. Finally, we conclude with a discussion in Section~\ref{discussion}. All proofs are provided in the \hyperref[appendix]{appendix}.

\section{Background on Partial Least Squares Regression}\label{section:scalar_pls}
This section provides a brief introduction to Nonlinear Iterative Partial Least Squares (NIPALS; \citealp{woldPathModelsLatent1975}), the standard algorithm for fitting PLS regression in Euclidean spaces. Consider the high-dimensional regression model:
\begin{equation*}
Y_i = \boldsymbol{\beta}^\top \mathbf{Z}_i + \epsilon_i, \quad i = 1, \ldots, n,
\end{equation*}
where $\mathbf{Z}_i \in \mathbb{R}^p$ represents the predictor vector, typically in the regime where the dimension $p$ far exceeds the sample size $n$ ($p \gg n$).
To address high multicollinearity in $\mathbf{Z}_i$, a common approach is to approximate $\mathbf{Z}_i$ using a low-dimensional score vector $\hat{\boldsymbol{\rho}}_i := (\widehat{\rho}_{i}^{[1]}, \ldots, \widehat{\rho}_i^{[L]})^\top \in \mathbb{R}^L$ with mutually uncorrelated entries, which serve as coefficients with respect to a chosen set of basis vectors (or directions). Ordinary least squares (OLS) is then performed on these scores.
While principal component analysis (PCA) constructs these directions to capture only the variation within the predictors, 
partial least squares (PLS) constructs them to explicitly maximize both predictor variation and the correlation with the response. 
 Let $\mathbf{Z} \in \mathbb{R}^{n \times p}$ denote the data matrix. 
 Assuming the data has been centered,
 the $l$-th PLS direction $\hat{\boldsymbol{\xi}}^{[l]}$ is defined as the solution to the constrained squared covariance maximization problem:
\begin{align*}
\max_{ \mathbf{ h } \in \mathbb{R}^p }&
\biggl( 
\frac{1}{n } \sum_{i=1}^n \langle \mathbf{h}, \mathbf{Z}_i   \rangle Y_i  
\biggr)^2
\numberthis \label{scalar_PLS_sample}
\\
\text{s.t.}&~\| \mathbf{ h } \|_2 = 1,~
\mathbf{h}^\top (\mathbf{Z}^\top \mathbf{Z}) \hat{\boldsymbol{\xi}}^{[j]} = 0, 
\quad j = 1, \ldots, l-1.
\end{align*}
The NIPALS algorithm,
summarized in Algorithm~\ref{alg:scalar_pls},
solves this by iteratively alternating between covariance maximization and residualization. 	\begin{algorithm}[b!] 
		\caption{Scalar partial least squares regression}\label{alg:scalar_pls}
		\begin{algorithmic}[1]
			\State \textbf{Input:} Standardized $\mathbf{Z}_1, \ldots, \mathbf{Z}_n$ and $Y_1, \ldots, Y_n$. 
			\For{$l = 1, 2, \ldots, L$}
			\\
			\textbf{PLS direction and score estimation:}
			\State $\widehat{\boldsymbol{\xi}}^{[l]} \leftarrow
			\arg \max_{  \boldsymbol{\alpha}   } 
            \biggl( 
\frac{1}{n } \sum_{i=1}^n \langle \boldsymbol{\alpha}, \mathbf{Z}^{[l]}_i   \rangle Y_i^{[l]}  
\biggr)^2
		~\text{s.t.}~\| \boldsymbol{\alpha} \|_2 = 1
			$
			\Comment{PLS direction}\label{alg:step:scalar_pls_direction}
			\State 
			$
			\widehat{\rho}_i^{[l]} \leftarrow \langle \hat{\boldsymbol{\xi}}^{[l]},   \mathbf{Z}^{[l]}_i \rangle, 
			~i=1,  \ldots, n
			$ \Comment{PLS score}
			\\
			\textbf{Residualization:}
			\State $\nu^{[l]} 
			\leftarrow
			\frac{
				\sum_{i=1}^n Y_i^{[l]} \widehat{\rho}_{i}^{[l]}}{
				\sum_{i=1}^n \widehat{\rho}_{i}^{[l]2}
			}$ \Comment{Least squares estimate} 
			\State $ Y_i^{[l+1]} \leftarrow Y_i^{[l]} - \nu^{[l]}\widehat{\rho}_i^{[l]}~i=1,  \ldots, n$
			\State $ \widehat{\boldsymbol{\delta}}^{[l]}  \leftarrow \frac{1}{\sum_{i=1}^n \widehat{\rho}_{i}^{[l]2}}\sum_{i=1}^n  \widehat{\rho}_{i}^{[l]}\mathbf{Z}_i^{[l]} $ \Comment{Least squares estimate} 
			\State $\mathbf{Z}_i^{[l+1]}  \leftarrow \mathbf{Z}_i^{[l]} -   \widehat{\rho}_{i}^{[l]}  \widehat{\boldsymbol{\delta}}^{[l]} $
			\EndFor
			\textbf{Regression coefficient estimation:}
			\State
			$\widehat{\boldsymbol{\iota}}^{[1]} \leftarrow \widehat{\boldsymbol{\xi}}^{[1]}$
			\For{$l =  2, \ldots, L$}
			\State $\widehat{\boldsymbol{\iota}}^{[l]} \leftarrow \widehat{ \boldsymbol{\xi} }^{[l]} - \sum_{u=1}^{l-1} \langle \widehat{\boldsymbol{\delta}}^{[u]}, \widehat{\boldsymbol{\xi}}^{[l]} \rangle 
			\widehat{\boldsymbol{\iota}}^{[u]}$  
			\EndFor
			$	\widehat{\boldsymbol{\beta}} \leftarrow \sum \limits_{l=1}^L\widehat{\nu}^{[l]} \widehat{\boldsymbol{\iota}}^{[l]}$ \Comment{regression coefficient estimate}
			\State \textbf{Output:} the  regression coefficient estimate
		\end{algorithmic}
	\end{algorithm}
%%%
Extending Algorithm \ref{alg:scalar_pls} to the hybrid model \eqref{PFLM} requires jointly handling correlated functional and scalar components. A naive pointwise approach, which concatenates functional evaluations and scalar predictors into a single Euclidean vector, is computationally prohibitive, incompatible with irregular data, and prone to overfitting.
 The extension to the hybrid setting faces three primary challenges :\begin{enumerate}\item Independent variables consist of multiple highly structured functional images and scalar predictors.\item Sample sizes are small relative to the high dimensionality of the combined predictors.\item Existing PLS methods typically accommodate only functional predictors \citep{predaPLSRegressionStochastic2005, delaigleMethodologyTheoryPartial2012} or a single functional predictor with scalar covariates \citep{Wang2018}, lacking a general framework for complex hybrid data.\end{enumerate}The next section proposes a new framework to address these limitations.

\section{Population-Level Hybrid PLS}\label{section:hybrid_hilbert}We now formulate the population-level PLS algorithm by defining a unified Hilbert space for hybrid predictors. This framework allows us to extend the classical NIPALS algorithm, which relies solely on inner product operations, to the hybrid setting.
\subsection{The Hybrid Hilbert Space}
To treat functional and scalar predictors within a single mathematical structure, we introduce the product space $\mathbb{H}$.
\begin{definition}[Hybrid Space]
\label{def:hilbert_space}
Let $\mathbb{H} := (\mathbb{L}^2[0,1])^K \times \mathbb{R}^p$.
An element $h \in \mathbb{H}$ is an ordered tuple $h = (f_1, \dots, f_K, \mathbf{u})$.
For any $h_1, h_2 \in \mathbb{H}$, we define the inner product:
\begin{equation}
\label{eq: hybrid inner product}
\langle h_1, h_2 \rangle_{\mathbb{H}} := \sum_{k=1}^K \int_0^1 f_{1k}(t) f_{2k}(t) \, dt 
+ 
\omega \mathbf{u}_1^\top \mathbf{u}_2,\end{equation}
where $\omega > 0$ is a weight parameter controlling the relative importance of the scalar component. This inner product induces the norm $\Vert h \Vert_{\mathbb{H}} := \langle h, h \rangle_{\mathbb{H}}^{1/2}$.
\end{definition}
The weight $\omega$ addresses the heterogeneity between functional and scalar parts, particularly regarding measurement scales (discussed further in Section \ref{subsec: Data Preprocessing}). 
For theoretical clarity, we assume $\omega=1$ in this section, noting that all results generalize to any $\omega > 0$.
Under this inner product, $\mathbb{H}$ is a separable Hilbert space (see Appendix \ref{section:proof:lemma:hybrid_hilbert_space} for details). Our method naturally generalizes to settings where each functional predictor resides in a unique Hilbert space, potentially defined over distinct compact domains in $\mathbb{R}^d$ with varying observation points. However, for notational simplicity, we assume a shared Hilbert space over $[0,1]$ for all functional predictors.

We define the hybrid predictor $W = (X_1, \ldots, X_K, \mathbf{Z})$ as a random element taking values in $\mathbb{H}$, measurable with respect to the Borel $\sigma$-field $\mathbbm{B}(\mathbb{H})$. The joint regression model \eqref{PFLM} can then be written concisely as:\begin{equation}\label{eq: Hybrid functional model}Y = \langle \beta, W \rangle_{\mathbb{H}} + \epsilon, \quad \text{where } \beta \in \mathbb{H}.\end{equation}
Throughout this section, we assume centered variables, $\mathbb{E}[Y] = 0$ and $\mathbb{E}[W] = 0$.

\subsection{Covariance Operators and PLS Directions}\label{section:operators}
The core of partial least squares (PLS) involves maximizing the squared covariance between the predictor and the response. 
To characterize this at the population level, we first define the covariance operator $\Sigma_W$ such that, for $u,v \in \mathbb{H}$, 
\begin{align*}
\Sigma_W u &= \mathbb{E}\big[ \langle W, u \rangle_{\mathbb{H}} \, W \big], \quad
\langle \Sigma_W u, v \rangle_{\mathbb{H}} = \mathbb{E}\big[ \langle W, u \rangle_{\mathbb{H}} \, \langle W, v \rangle_{\mathbb{H}} \big].
\end{align*}
Next, we define the cross-covariance vector 
\[
\Sigma_{YW} := \mathbb{E}[\, Y W \,] \in \mathbb{H},
\] 
which is composed of functional components  and scalar components 
\[
\sigma_{YX} := \big( \mathbb{E}[Y X_1], \ldots, \mathbb{E}[Y X_K] \big)
\quad
\sigma_{YZ} := \big( \mathbb{E}[Y Z_1], \ldots, \mathbb{E}[Y Z_p] \big)^\top.
\] 
This, in turn, induces the cross-covariance operator 
$
\mathcal{C}_{YW} : \mathbb{H} \to \mathbb{R}$, defined as
$\mathcal{C}_{YW} h := \langle \Sigma_{YW}, h \rangle_\mathbb{H}.
$
Finally, we define the composite operator $\mathcal{U}: \mathbb{H} \to \mathbb{H}$ as $\mathcal{U} := \Sigma_{YW} \otimes \Sigma_{YW}$, which acts on any $h \in \mathbb{H}$ via:
		\begin{equation*}
			\mathcal{U} h 
			= 
			\mathcal{C}_{WY} ( \mathcal{C}_{YW} h) 
			=
			\mathcal{C}_{YW} (\langle \Sigma_{YW}, h \rangle_\mathbb{H}) 
			=  
			\Sigma_{YW} \, \langle \Sigma_{YW}, h \rangle_\mathbb{H} 
			= 
			(\Sigma_{YW} \otimes \Sigma_{YW}) h.
		\end{equation*}
To ensure the constrained squared covariance maximization problem in NIPALS is well-posed in the hybrid space, we show that, under mild moment conditions, $\mathcal{U}$ is a compact, self-adjoint, and positive-semidefinite operator.  Proof of Lemma \ref{lemma:cross_cov_functional}
	is provided in Appendix \ref{section:proof:lemma:cross_cov_functional}.
\begin{lemma}\label{lemma:cross_cov_functional}
If   there exist finite constants $Q_1$ and $Q_2$ such that
		\begin{align}\label{condition_for_compact}
			\max \limits_{k=1, \ldots, K} \sup \limits_{t \in [0,1]} \mathbb{E}
[Y X_{k}](t)^2 < Q_1 \quad \textnormal{and} \quad  \max_{r=1, \ldots, p} \mathbb{E}
[Y Z_{r}]^2 < Q_2.
		\end{align}the operator
         $\mathcal{U}$
		is  self-adjoint, positive-semidefinite  and compact.
	\end{lemma}

By the Hilbert-Schmidt theorem (e.g., \citealp[Theorem 4.2.4]{hsingTheoreticalFoundationsFunctional2015}), Lemma \ref{lemma:cross_cov_functional} guarantees that $\mathcal{U}$ admits a complete orthonormal system of eigenfunctions $\{\xi_{(u)}\}_{u \in \mathbb{N}}$ with corresponding non-negative eigenvalues $\kappa_{(1)} \ge \kappa_{(2)} \ge \cdots \to 0$.
The population PLS algorithm constructs directions by maximizing squared covariance. The following Theorem establishes that this optimization is equivalent to an eigenproblem for $\mathcal{U}$, ensuring its well-posedness via Lemma \ref{lemma:cross_cov_functional}.
\begin{theorem} \label{theorem:population_PLS}
At the $l$-th step, the population PLS direction $\xi^{[l]}$ is defined as the maximizer of the squared covariance:
\begin{equation}\label{const_eigen_pop}
\xi^{[l]} := \arg \max_{h \in \mathbb{H}, \|h\|_{\mathbb{H}}=1} \operatorname{Cov}^2(\langle W^{[l]}, h \rangle_{\mathbb{H}}, Y^{[l]}).
\end{equation}
Under the condition of Lemma \ref{lemma:cross_cov_functional},
this maximum is attained by the eigenfunction corresponding to the largest eigenvalue of the operator $\mathcal{U}^{[l]} := \Sigma_{Y^{[l]}W^{[l]}} \otimes \Sigma_{Y^{[l]}W^{[l]}}$.
\end{theorem}The proof of Theorem \ref{theorem:population_PLS} is provided in Appendix \ref{section:proof:theorem:population_PLS}.
Based on this legitimate eigenproblem, we extend NIPALS to the hybrid setting as follows:\begin{theorem}[Properties of Population PLS]\label{prop:residualization_equiv_eigen}The $l$-th PLS direction $\xi^{[l]}$, solving \eqref{const_eigen_pop}, equivalently satisfies the constrained maximization:\begin{equation*}
    \xi^{[l]} := \arg\max_{h \in \mathbb{H}}
    \operatorname{Cov}^2(\langle W, h \rangle_{\mathbb{H}}, Y),~s.t.~\|h\|_{\mathbb{H}}=1, \langle h, \Sigma_W \, \xi^{[j]} \rangle_{\mathbb{H}} = 0,\, j < l.
\end{equation*}
Furthermore, the algorithm guarantees orthogonality of the scores and uncorrelated residuals. For every \(s\ge 1\) and every \(j\le s-1\), we have
$
\mathbb{E}[\rho^{[j]}\rho^{[s]}]=0
$
and
$
\mathbb{E}[Y^{[s]}\rho^{[j]}]=0.
$\end{theorem}
 Proof of Proposition \ref{prop:residualization_equiv_eigen} is provided in Appendix \ref{section:proof:prop:residualization_equiv_eigen}.
  This theoretical result legitimizes the iterative NIPALS algorithm in the hybrid space. The procedure is summarized in Algorithm \ref{alg:population_pls}.	\begin{algorithm}[t!]
		\caption{Population NIPALS}\label{alg:population_pls}
		\begin{algorithmic}[1]
        \State
          $(W^{[1]}, Y^{[1]}) \leftarrow (W, Y)$
			\For{$l = 1, 2, \ldots, L$}
\State
$
\xi^{[l]}
\leftarrow
\arg\max_{h \in \mathbb{H}} \operatorname{Cov}^2(\langle W^{[l]}, h \rangle_{\mathbb{H}}, Y^{[l]})~s.t.~ \|h\|_{\mathbb{H}} = 1,
$\label{alg:line:pls_direction} \Comment{PLS direction}
%%%%%%%%%
\State 
$
\rho^{[l]} \leftarrow \langle \xi^{[l]}, W^{[l]} \rangle 
$ \Comment{PLS score}
%%%%%%%%%%
\State
$
\delta^{[l]}
\leftarrow
\frac{1}{\mathbb{E}[ (\rho^{[l]})^2 ]}  \mathbb{E}[W^{[l]} \rho^{[l]}]
$\label{population_delta}
\Comment{Linear regression of 
$W^{[l]}$
on $\rho^{[l]}$}
%%%%%%%%%%%
\State
$
\nu^{[l]}
\leftarrow
\frac{1}{\mathbb{E}[ (\rho^{[l]})^2 ]}
\mathbb{E}[Y^{[l]} \rho^{[l]}]
$
\Comment{
Linear regression of 
$Y^{[l]} $
on $\rho^{[l]}$
}
%%%%%%%%%%%%%%%%%%%
\State 
$ W^{[l+1]} \leftarrow W^{[l]} - \rho^{[l]} \, \delta^{[l]}$
\label{population_predictor_residualize}
\Comment{
Residualized predictor
}
\State
$ Y^{[l+1]} \leftarrow Y^{[l]} - \nu^{[l]} \, \rho^{[l]}$
\Comment{
Residualized response
}
\EndFor
\State \textbf{Output:}
PLS directions $\xi^{[1]}, \ldots, \xi^{[L]}$
		\end{algorithmic}
	\end{algorithm}

  Finally, the following theorem establishes that the fitted values obtained from hybrid PLS converge to the original response in the mean-squared sense. This result demonstrates that hybrid PLS, as a dimension-reduction technique, does not lose information about the response when a sufficient number of components is retained, thereby distinguishing it from PCA.
\begin{theorem}\label{thm:hybrid_pls_convergence}
	  The Hybrid PLS fitted value converges to $Y$ in the mean-squared sense:
	\begin{equation*}
		\lim_{L \to \infty} \mathbb{E}
        \biggl[ \bigl\| \sum_{l=1}^L \nu^{[l]} \rho^{[l]} - Y  \bigr\|_2^2 \biggr] = 0.
	\end{equation*}
\end{theorem}
 The proof of Theorem \ref{thm:hybrid_pls_convergence} is provided in Appendix \ref{section:proof:thm:hybrid_pls_convergence}.	

   \section{Proposed PLS Algorithm}\label{section:main:our_algorithm}

We propose a sample version of hybrid NIPALS (Algorithm \ref{alg:hybrid_pls}) to efficiently compute PLS components for multiple dense or irregular functional predictors alongside scalar predictors. The algorithm alternates between two subroutines: regularized component estimation (Section \ref{section:sub:compute_PLS_component}) and residualization (Section \ref{section:sub:residualization}).

\subsection{Preliminary step: Finite-basis approximation, notations, and normalization}\label{sec: finite basis approximation}
Let $\{b_m(t)\}_{m=1}^M$ denote a basis of $\mathbb{L}^2([0,1])$ with linearly independent second derivatives (e.g., cubic B-splines, Fourier series, or orthonormal polynomials). We employ this basis,  truncated at a moderate size (e.g., $M \approx 20$), to approximate the functional predictors $X_{ij}(t)$, regression coefficients $\beta_j(t)$, PLS directions $\xi_j(t)$, and residualization coefficients $\delta_j(t)$; iteration indices are suppressed here for clarity.
Consequently, all computations are performed within the finite-dimensional approximation space $\widetilde{\mathbb{H}} := \operatorname{span}(b_1, \ldots, b_M)^K \times \mathbb{R}^p$. We denote the approximated predictors for the $i$-th subject as $\widetilde{W}_i = (\widetilde{X}_{i1}, \ldots, \widetilde{X}_{iK}, \mathbf{Z}_i)$. 
We denote the scalar basis coefficients for these functions as $\{\theta_{ijm}\}, \{ \eta_{jm} \}, \{ \gamma_{jm} \},$ and $\{ \pi_{jm} \}$, respectively. Finally, we collect these coefficients into the corresponding $M$-dimensional vectors $\boldsymbol{\theta}_{ij}, \boldsymbol{\eta}_j, \boldsymbol{\gamma}_j, \text{ and } \boldsymbol{\pi}_j$.

\paragraph{Matrix tuple Representation}~We organize the predictor coefficients into matrices $\Theta_j := (\boldsymbol{\theta}_{1j}, \ldots, \boldsymbol{\theta}_{nj})^\top \in \mathbb{R}^{n \times M}$ and stack the full set of predictors as $\Theta := (\Theta_1, \ldots, \Theta_K, \mathbf{Z}) \in \mathbb{R}^{n \times (MK + p)}$. The response vector is denoted by $\mathbf{y} := (y_1, \ldots, y_n)^\top$.Next, let $\mathbf{B}, \mathbf{B}^{\prime\prime} \in \mathbb{R}^{M \times M}$ be the Gram matrices containing the inner products of the basis functions and their second derivatives, defined by entries $B_{m, m'} := \int_0^1 b_m(t) b_{m'}(t) dt$ and $B^{\prime\prime}_{m, m'} := \int_0^1 b_m''(t) b_{m'}''(t) dt$. We construct the corresponding block-diagonal matrices:\begin{equation}\label{def:gram_block}\mathbb{B} := \operatorname{blkdiag}(\mathbf{B}, \ldots, \mathbf{B}, \mathbf{I}_p), \quad\mathbb{B}^{\prime\prime} := \operatorname{blkdiag}(B^{\prime\prime}, \ldots, B^{\prime\prime}, \mathbf{I}_p).\end{equation}The data at the $l$-th iteration (index omitted) is thus fully characterized by the tuple\begin{equation}\label{def:problem_instance}(\mathbb{B}, \mathbb{B}^{\prime\prime}, \Theta, \mathbf{y}) \in\mathbb{R}^{(MK+p) \times (MK+p)} \times\mathbb{R}^{(MK+p) \times (MK+p)} \times\mathbb{R}^{n \times (MK+p)} \times\mathbb{R}^n.\end{equation}

Although we employ a common basis for simplicity, the definitions of $\mathbb{B}$ and $\mathbb{B}^{\prime\prime}$ readily generalize to distinct bases for each predictor, facilitating the analysis of multiple dense or irregular functional predictors.

\paragraph{Data Preprocessing} \label{subsec: Data Preprocessing}
To handle discrepancies in units and variation between functional and scalar predictors, we employ a two-step standardization. First, functional predictors are standardized to zero mean and unit integrated variance, and scalar predictors to zero mean and unit variance. Second, to balance the influence of the functional and scalar parts, we scale the scalar vector $\mathbf{Z}_i$ by a factor $\omega^{1/2}$, where
\begin{equation*}
    \omega = \frac{\sum_{i=1}^n \sum_{k=1}^K \Vert \widetilde{X}_{ik} \Vert^2_{\mathbb{L}^2}}{\sum_{i=1}^n \Vert \mathbf{Z}_i \Vert_2^2}.
    \label{eq: weight}
\end{equation*}
This is equivalent to using a weighted inner product \eqref{eq: hybrid inner product} in the hybrid space.

\subsection{Iterative Steps}\label{section:sub:iterative}
The algorithm iteratively constructs an orthonormal hybrid basis that captures predictor-response relationships. Each iteration consists of two steps: estimating the PLS direction (Section \ref{section:sub:compute_PLS_component}) and residualizing the data (Section \ref{section:sub:residualization}). iteration indices are suppressed where unambiguous.

\subsubsection{Step 1: Regularized estimation of PLS direction} \label{section:sub:compute_PLS_component}
Since the regression coefficient $\hat{\beta}$ is derived as a linear combination of PLS directions (Section \ref{section:sub:regression_coeff}), non-smooth PLS directions inherently yield a non-smooth $\hat{\beta}$, leading to potential overfitting and reduced interpretability. To address this, we estimate the PLS direction $\xi \in \widetilde{\mathbb{H}}$ by maximizing the squared empirical covariance with the response, subject to a roughness constraint. Adopting the generalized smoothing framework  \citep{silvermanSmoothedFunctionalPrincipal1996}, this constraint regulates the complexity of the functional components while facilitating the borrowing of strength across them. Accordingly, at the $l$-th iteration (with indices suppressed and assuming prior residualization), our procedure solves the following optimization problem:
\begin{equation}\label{def:maximizer_squared_empirical_cov_reg}
    \hat{\xi}
    := \arg \max_{\xi \in \widetilde{\mathbb{H}}
    }~\widehat{\textnormal{Cov}}^2
    \bigl(
    \langle \widetilde{W}, \xi \rangle_{\mathbb{H}}, Y
    \bigr)\quad \text{s.t.} \quad \| \xi  \|_\mathbb{H}^2 +  \sum \limits_{j=1}^K \lambda_j
    \int_0^1 \bigl\{ \xi_j^{\prime\prime}(t) \bigr\}^2 dt = 1.
\end{equation}
The smoothing parameters $\{\lambda_j\}$ control the trade-off between covariance maximization and curve smoothness; notably, setting $\lambda_j = 0$ recovers the standard unregularized PLS solution. The resulting solution $\hat{\xi} \in \widetilde{\mathbb{H}}$ constitutes an ordered pair expanded as:
\begin{equation*}
\hat{\xi}   =
\bigl(
\hat{\xi}_1(t), \ldots, \hat{\xi}_K(t), \hat{\boldsymbol{\zeta}}
\bigr)
=
\biggl(
\sum_{m=1}^M \hat{\gamma}_{1m} , b_m(t),
\ldots,
\sum_{m=1}^M \hat{\gamma}_{Km} , b_m(t), \hat{\boldsymbol{\zeta}}
\biggr).
\end{equation*}
We define the correlation matrix
\[
\mathbf{V} := n^{-2} \big( \mathbb{B} \Theta^\top \mathbf{y} \big) \big( \mathbb{B} \Theta^\top \mathbf{y} \big)^\top,
\] 
and the penalty matrix
\[
\Lambda := \operatorname{blkdiag}\big( \lambda_1 \mathbf{I}_M, \ldots, \lambda_K \mathbf{I}_M, \mathbf{0}_{p \times p} \big).
\]
The optimization in \eqref{def:maximizer_squared_empirical_cov_reg} is equivalent to the generalized Rayleigh quotient:
	\begin{theorem}[Regularized estimation of PLS component direction]\label{proposition:eigen_regul}
		Let
		$
		(\mathbb{B}, \mathbb{B}^{\prime \prime}, \Theta, \mathbf{y})
		$
		denote the  data given at the $l$-th iteration, as defined in \eqref{def:problem_instance}.
		$\Lambda \in \mathbb{R}^{(MK+p) \times (MK+p)}$ be defined as:
		\begin{equation}\label{def:Lambda}
			\Lambda := \operatorname{blkdiag}(\lambda_1 \mathbf{I}_M, \ldots, \lambda_K \mathbf{I}_M, \mathbf{0}_{p \times p}),
			~\text{where}~\lambda_1, \ldots, \lambda_K \geq 0.
		\end{equation}
		Here, $\mathbf{0}_{p \times p}$ denotes the $p \times p$ zero matrix.
		The coefficients of the squared covariance  maximizer defined in \eqref{def:maximizer_squared_empirical_cov_reg}, 
		are obtained as
		\begin{equation}	\label{eq: Regularized generalized rayleigh quotient equation}
			\left(
			\hat{\gamma}_{11}, \ldots, \hat{\gamma}_{1M}
			, \ldots,
			\hat{\gamma}_{K1}, \ldots, \hat{\gamma}_{KM}
			, \hat{\boldsymbol{\zeta}}^\top
			\right)^\top
			\hspace{-.7em}
			=
			\arg 
			\hspace{-.6em}
			\max_{\boldsymbol{\xi} \in \mathbb{R}^{MK+p}}   \boldsymbol{\xi}^\top \mathbf{V} \boldsymbol{\xi}
			~~\text{s.t.}~~
			%
			% constraint
			\boldsymbol{\xi}^\top 
			\hspace{-.2em}
			(\mathbb{B} + \Lambda \mathbb{B}^{\prime \prime}) \boldsymbol{\xi} = 1.
		\end{equation}
		
	\end{theorem}
	The proof of Proposition \ref{proposition:eigen_regul} is provided in Appendix  \ref{section:proof:proposition:eigen_regul}. The constraint $\boldsymbol{\xi}^\top (\mathbb{B} + \Lambda \mathbb{B}^{\prime \prime}) \boldsymbol{\xi} = 1$ enforces orthonormality under a modified inner product (Section \ref{section:sub:geom}). The parameter $\lambda_k$ governs the smoothness-fit trade-off: setting $\lambda_k=0$ recovers the unregularized solution (Proposition \ref{proposition:eigen_noregul}), while $\lambda_k \to \infty$ forces a linear structure. Values of $\{\lambda_k\}$   are selected via cross-validation.

 While generalized eigenproblems can be computationally unstable, the rank-one structure of $V$ allows for an efficient closed-form solution solving linear systems for the functional and scalar components separately, followed by normalization by a common factor. 

\begin{theorem}[Closed-form solution]  \label{proposition:linear_regul}
Let us define 
\begin{equation*}
\mathbf{u}_j := \mathbf{B} \Theta_j^\top \mathbf{y} \in \mathbb{R}^M,~j=1,\ldots,K,
\quad 
\mathbf{v} := \mathbf{Z}^\top \mathbf{y} \in \mathbb{R}^p,
\quad
 q := \sum_{j=1}^K \mathbf{u}_j^\top (\mathbf{B} + \lambda_j \mathbf{B}^{\prime\prime})^{-1} \mathbf{u}_j + \mathbf{v}^\top \mathbf{v}.
\end{equation*}
    The unique solution (up to sign) to \eqref{eq: Regularized generalized rayleigh quotient equation} is:
    \[
    \hat{\boldsymbol{\gamma}}_j = \frac{1}{\sqrt{q}} (\mathbf{B} + \lambda_j \mathbf{B}^{\prime\prime})^{-1} \mathbf{u}_j, \quad j = 1, \ldots, K; \qquad
    \hat{\boldsymbol{\zeta}} = \frac{1}{\sqrt{q}} \mathbf{v}.
    \]
\end{theorem}
The proof of Proposition~\ref{proposition:linear_regul} is provided in Appendix~\ref{section:proof:proposition:linear_regul}.
The normalization factor $q$ couples the functional and scalar components, ensuring the solution captures the correlation structure between them.

\subsubsection{Iterative step 2: residualization via hybrid-on-scalar regression} \label{section:sub:residualization}

The second step of the $l$-th iteration residualizes both predictors and responses. First, we compute the PLS scores using the estimated direction $\widehat{\xi}^{[l]}$:
\begin{equation}\label{def:plsscore}
    \hat{\rho}^{[l]}_i := \langle \widetilde{W}_i^{[l]}, \, \hat{\xi}^{[l]} \rangle_{\mathbb{H}}.
\end{equation}
Assuming centered data, these scores possess a sample mean of zero. Next, we regress the current predictors and responses on these scores. The hybrid regression coefficient $\hat{\delta}^{[l]} \in \widetilde{\mathbb{H}}$ is estimated via unpenalized least squares:
\begin{equation}\label{def:argmin_pensse}
    \hat{\delta}^{[l]} := \arg \min_{\delta \in \widetilde{\mathbb{H}}} \sum_{i=1}^n 
    \| \widetilde{W}_i^{[l]} - \hat{\rho}_i^{[l]}\, \delta\|_{\mathbb{H}}^2.
\end{equation}
We omit roughness penalties in this step to ensure the orthonormality of the resulting PLS components (see Section \ref{section:sub:regression_coeff}), relying instead on the smoothness of $\hat{\xi}^{[l]}$ to regulate the regression coefficient $\beta$.
Simultaneously, the response regression coefficient $\hat{\nu}^{[l]}$ is estimated via scalar least squares. The following lemma provides the closed-form updates, which mirror standard scalar PLS.

\begin{lemma}[Closed-form solution]\label{proposition:closed_form_orthgonalization}
    Let $\hat{\boldsymbol{\rho}}^{[l]} := (\hat{\rho}^{[l]}_1, \ldots, \hat{\rho}^{[l]}_n)^\top$. The predictors and responses for the $(l+1)$-th iteration are updated as:
    \begin{equation*}
        \widetilde{W}_i^{[l+1]} 
        :=  \widetilde{W}_i^{[l]}  - 
        \widehat{\rho}_i^{[l]}
        \hat{\delta}^{[l]}, \quad \text{where} \quad
        \hat{\delta}^{[l]}
        := 
        \frac{1}{\|   \hat{\boldsymbol{\rho}}^{[l]}\|_2^2}
        \sum_{i=1}^n 
        \widehat{\rho}_i^{[l]}
        \widetilde{W}_i^{[l]},
    \end{equation*}     
    and
    \begin{equation}\label{algorithm_step:residualization_predictor}
        Y_i^{[l+1]} = Y_i^{[l]} - 
        \hat{\nu}^{[l]} 
        \widehat{\rho}_i^{[l]}, \quad \text{where} \quad \hat{\nu}^{[l]} :=  
        \frac{
            \mathbf{y}^{[l] \top }  \hat{\boldsymbol{\rho}}^{[l]}
        }{
            \| \hat{\boldsymbol{\rho}}^{[l]} \|_2^2
        }.
    \end{equation}
\end{lemma}
Proof of Lemma \ref{proposition:closed_form_orthgonalization} is provided in Appendix \ref{section:proof:proposition:closed_form_orthgonalization}. Note that since the inputs are centered, the residuals $\widetilde{W}_i^{[l+1]}$ and $Y_i^{[l+1]}$ remain centered.

\subsection{Final step: estimating  the hybrid regression coefficient}
	\label{section:sub:regression_coeff}
	The hybrid regression coefficient $\beta$ in model \eqref{eq: Hybrid functional model} can be written 
	as a linear combination of PLS directions:
	\begin{lemma}\label{lemma:recursive}
		Let us define
		$
		\widehat{\iota}^{[1]} := \widehat{\xi}^{[1]}.
		$
		For $l \ge 2$, we recursively   define:
		\begin{equation*}
			\widehat{ \iota}^{[l]} = \widehat{\xi}^{[l]} - \sum_{u=1}^{l-1} \langle \widehat{ \delta}^{[u]}, \widehat{\xi}^{[l]} \rangle_{\mathbb{H}} \widehat{ \iota}^{[u]}. 
		\end{equation*}
		Then we have:
		\begin{equation*}
			\widehat{\rho}_i^{[l]}
			=
			\langle  W_i^{[l]}, \widehat{\xi}^{[l]} \rangle_\mathbb{H}
			=
			\langle  W_i, \widehat{ \iota}^{[l]} \rangle_\mathbb{H}.
		\end{equation*}
	\end{lemma}
	Proof of Lemma \ref{lemma:recursive} is provided in Appendix \ref{section:proof:lemma:recursive}.
	
	Next,  \eqref{algorithm_step:residualization_predictor} leads to the following model:
	\begin{equation*}
		Y_i = \sum \limits_{l=1}^L \widehat{\nu}^{[l]} \widehat{\rho}_i^{[l]} + \epsilon_i.
	\end{equation*}
	This model lets us to express $Y_i$ as:
	\begin{equation*}
		Y_i = \sum \limits_{l=1}^L \widehat{\nu}^{[l]} \langle W_i^{[l]}, \widehat{\xi}^{[l]} \rangle_\mathbb{H} +\epsilon_i 
		= 
		\bigl \langle W_i, \sum \limits_{l=1}^L \widehat{\nu}^{[l]} \widehat{ \iota }^{[l]}
		\bigr \rangle_\mathbb{H}+\epsilon_i,
	\end{equation*}
	which, given the uniqueness of $ \beta $, leads to
	\begin{equation}\label{beta_form}
		\widehat{ \beta} = \sum \limits_{l=1}^L\widehat{\nu}^{[l]} \widehat{ \iota }^{[l]}
		= 
		\sum_{l=1}^L \left( \widehat{\nu}^{[l]} - \sum_{k=l+1}^L \widehat{\nu}^{[k]} \langle \widehat{\delta}^{[l]}, \widehat{\xi}^{[k]} \rangle_{\mathbb{H}} \right) \widehat{\xi}^{[l]}.
	\end{equation}

	\begin{algorithm}[t!] 
		\caption{Hybrid partial least squares regression}\label{alg:hybrid_pls}
		\begin{algorithmic}[1]
			\State 	\textbf{Initialize:} $(\mathbb{B}, \mathbb{B}^{\prime\prime}, \Theta^{[1]}, \mathbf{y}^{[1]})$ as the data objects after basis expansion (Section \ref{sec: finite basis approximation}).
			
			\State   $\widetilde{W}_1^{[1]}, \ldots, \widetilde{W}_n^{[1]}, Y_1^{[1]}, \ldots, Y_n^{[1]} \leftarrow$ standardized versions of $W_1, \ldots, W_n, Y_1, \ldots, Y_n$, following Section \ref{subsec: Data Preprocessing}
			\For{$l = 1, 2, \ldots, L$}
			\\
			\textbf{PLS direction and score estimation  (Proposition \ref{proposition:linear_regul}): }
			\State $\mathbf{u}_j^{[l]} \leftarrow B \Theta_j^{[l]\top} \mathbf{y}^{[l]}, ~j = 1, \ldots, K$
			\State $\mathbf{v}^{[l]} \leftarrow \mathbf{Z}^{[l]\top} \mathbf{y}^{[l]} $
			\State $	q^{[l]}  \leftarrow \sum_{j=1}^K \mathbf{u}_j^{[l]\top} (\mathbf{B} + \lambda_j \mathbf{B}^{\prime\prime})^{-1} \mathbf{u}_j^{[l]} + \mathbf{v}^{[l]\top} \mathbf{v}^{[l]}$
			\State $( \hat{\gamma}_{j1}^{[l]} , \ldots, \hat{\gamma}_{jM}^{[l]} )^\top  \leftarrow \frac{1}{\sqrt{q}} (\mathbf{B} + \lambda_j \mathbf{B}^{\prime\prime})^{-1} \mathbf{u}_j^{[l]} , ~j = 1, \ldots, K$
			\State $\hat{\boldsymbol{\zeta}}^{[l]}  \leftarrow \frac{1}{\sqrt{q}} \mathbf{v}^{[l]} $
			\State $	\hat{\xi}^{[l]}	\leftarrow 
			\biggl(
			\sum_{m=1}^M \hat{\gamma}_{1m}^{[l]} \, b_m(t),
			\ldots, 
			\sum_{m=1}^M \hat{\gamma}_{Km}^{[l]} \, b_m(t), \hat{\boldsymbol{\zeta}}^{[l]}
			\biggr)
			$
			\Comment{PLS direction} 
			\State 
			$
			\widehat{\rho}_i^{[l]} \leftarrow \langle \hat{\xi}^{[l]},   \tilde{W}^{[l]}_i \rangle, 
			~i=1,  \ldots, n
			$ \Comment{PLS score}
			\\
			\textbf{Residualization (Proposition \ref{proposition:closed_form_orthgonalization}):}
			\State $\nu^{[l]} 
			\leftarrow
			\frac{
				\sum_{i=1}^n Y_i^{[l]} \widehat{\rho}_{i}^{[l]}}{
				\sum_{i=1}^n \widehat{\rho}_{i}^{[l]2}
			}$ \Comment{Least squares estimate} 
			\State $ Y_i^{[l+1]} \leftarrow Y_i^{[l]} - \nu^{[l]}\widehat{\rho}_i^{[l]}~i=1,  \ldots, n$
			\State $ \widehat{ \delta }^{[l]}  \leftarrow \frac{1}{\sum_{i=1}^n \widehat{\rho}_{i}^{[l]2}}\sum_{i=1}^n  \widehat{\rho}_{i}^{[l]}\widetilde{W}_i^{[l]} $ \Comment{Least squares estimate} 
			\State $\widetilde{W}_i^{[l+1]}  \leftarrow \widetilde{W}_i^{[l]} -   \widehat{\rho}_{i}^{[l]}  \widehat{ \delta}^{[l]} $
			\EndFor
			\textbf{Regression coefficient estimation (Section \ref{section:sub:regression_coeff}):}
			\State
			$\widehat{ \iota }^{[1]} \leftarrow \widehat{ \xi }^{[1]}$
			\For{$l =  2, \ldots, L$}
			\State $\widehat{ \iota }^{[l]} \leftarrow \widehat{  \xi }^{[l]} - \sum_{u=1}^{l-1} \langle \widehat{ \delta }^{[u]}, \widehat{ \xi }^{[l]} \rangle  \widehat{ \iota }^{[u]}$  
			\EndFor
			$	\widehat{ \beta } \leftarrow \sum \limits_{l=1}^L\widehat{\nu}^{[l]} \widehat{ \iota}^{[l]}$ \Comment{regression coefficient estimate}
			\State \textbf{Output:} the  regression coefficient estimate $	\widehat{ \beta }$
		\end{algorithmic}
	\end{algorithm}

	\section{Geometric Properties of the hybrid PLS}
\label{section:sub:geom}
	This section provides the mathematical properties that support the algorithm suggested in Section \ref{section:main:our_algorithm}. Specifically,
	it demonstrates that our algorithm preserves the core properties of PLS, namely the orthonormality of the derived directions and the orthogonality of the scores.
	
	A key property of PLS is that its directions are orthonormal and its scores orthogonal across iterations. Our regularized estimates preserve this property under a modified inner product that incorporates the roughness penalty, defined as follows:
	\begin{definition}[Roughness-sensitive inner product]\label{def:hybrid_inner_product_roughness}
		Given two hybrid predictors $W_1 = (X_{11}, \ldots, X_{1K}, \mathbf{Z}_1)$ and $W_2 = (X_{21}, \ldots, X_{2K}, \mathbf{Z}_2)$, both elements of $\mathbb{H}$ as defined in Definition \ref{def:hilbert_space}, and a roughness penalty matrix $\Lambda = \operatorname{blkdiag}(\lambda_1 \mathbf{I}_M, \ldots, \lambda_K \mathbf{I}_M, \mathbf{0}_{p \times p})$, the roughness-sensitive inner product between $W_1$ and $W_2$ is defined as:
		\begin{equation}		\label{eq: hybrid inner product_roughness}
			\langle W_1, W_2\rangle_{\mathbb{H}, \Lambda}
			:=
			\sum \limits_{k=1}^K \int_0^1 X_{1k}(t) X_{2k}(t) \, dt
			+
			\sum \limits_{k=1}^K \lambda_k \int_0^1 X^{\prime \prime}_{1k}(t) X^{\prime \prime}_{2k}(t) \, dt
			+
			\mathbf{Z}_1^\top \mathbf{Z}_2.
		\end{equation}
	\end{definition}
	Based on this inner product, the following proposition states that the PLS component directions estimated from Proposition  \ref{proposition:eigen_regul} are orthonormal.
	\begin{theorem}[Orthonormality of estimated PLS component directions]\label{proposition: modified orthnormality of PLS components}
		The PLS component directions $
		\widehat{\xi}^{[1]}, \widehat{\xi}^{[2]}, \ldots, \widehat{\xi}^{[L]}
		$, estimated via Proposition \ref{proposition:eigen_regul} with a roughness penalty matrix $\Lambda = \operatorname{blkdiag}(\lambda_1 \mathbf{I}_M, \ldots, \lambda_K \mathbf{I}_M, \mathbf{0}_{p \times p})$, are mutually orthonormal with respect to the inner product $\langle \cdot, \cdot \rangle_{\mathbb{H}, \Lambda}$. That is,
		\begin{equation*}
			\langle \widehat{\xi}^{[l_1]}, \widehat{\xi}^{[l_2]} \rangle_{\mathbb{H}, \Lambda}
			= \mathbbm{1}(l_1 = l_2), \quad l_1, l_2 = 1, \ldots, L.
		\end{equation*}
	\end{theorem}
	The proof of Proposition \ref{proposition: modified orthnormality of PLS components} is provided in Appendix \ref{section:proof:proposition: modified orthnormality of PLS components}.
	
	The next proposition states that the vectors of estimated PLS scores for different iteration numbers are mutually orthogonal.
	\begin{theorem}	\label{proposition: orthnormality of PLS scores}
		Recall from Lemma \ref{proposition:closed_form_orthgonalization} that   $\widehat{\boldsymbol{\rho}}^{[l]}$ denote the $n$-dimensional vector whose elements consist of the $l$-th estimated PLS scores ($l=1, \ldots, L$) of $n$ observations.
		The vectors $\widehat{\boldsymbol{\rho}}^{[1]}, \widehat{\boldsymbol{\rho}}^{[2]}, \ldots, \widehat{\boldsymbol{\rho}}^{[L]}$  are mutually orthogonal in the sense that
		\begin{equation*}
			\widehat{\boldsymbol{\rho}}^{[l_1] \top}  \widehat{\boldsymbol{\rho}}^{[l_2] }= 0 \quad \textnormal{for} \quad \; l_1,l_2 \in \{1, \ldots, L\}, \; l_1 \ne l_2. 
		\end{equation*}
	\end{theorem}
	The proof of Proposition \ref{proposition: orthnormality of PLS scores} is provided in Appendix \ref{section:proof:proposition: orthnormality of PLS scores}.

\section{Numerical Studies}\label{section:experiments}
This section evaluates the finite-sample performance of the proposed Hybrid PLS algorithm. We begin in Section \ref{section:simul_geom} by  numerically verifying the key algebraic properties derived in Section \ref{section:main:our_algorithm}. Section \ref{section:simul:beta} subsequently assesses the sensitivity of regression coefficient estimation to regularization. In Section \ref{section:simul_comparison}, we utilize synthetic data to demonstrate the method's capacity to isolate predictive signals amidst non-predictive noise and cross-modality correlation. Finally, Section \ref{section:real_data} establishes the method superiority over baseline methods through an application to the Emory University renal scan data.

Throughout these analyses, we consider settings involving two functional predictors supported on $t \in [0,1]$, denoted as $X_1(t)$ and $X_2(t)$. Unless otherwise specified, the corresponding regularization parameters $\lambda_1$ and $\lambda_2$ are selected via 5-fold cross-validation on the training set, minimizing the root mean prediction squared error (RMSE) over the grid $\{0.001, 0.01, 0.1,   1\}$. All results are reported based on 100 Monte Carlo replications. The proposed method is implemented in the \texttt{FShybridPLS} R package. Source code and replication scripts for the simulations can be found at \url{https://github.com/Jong-Min-Moon/FShybridPLS}.

\subsection{Numerical Validation of Geometric Properties}\label{section:simul_geom}
We numerically verify the geometric properties of hybrid PLS established in Section~\ref{section:sub:geom}, namely orthogonality between residualized predictors and extracted directions (Lemma~\ref{lemma:annihilation}), mutual orthogonality of the PLS directions (Proposition~\ref{proposition: modified orthnormality of PLS components}), and orthogonality of the resulting PLS scores (Proposition~\ref{proposition: orthnormality of PLS scores}).

Synthetic data are generated to stress-test these algebraic properties, which are independent of distributional assumptions and predictive structure. Two independent latent factors are sampled as
$
U_{i} \sim \mathcal{N}(0,10^2)$ 
and
$V_{i} \sim \mathcal{N}(0,0.1^2).
$
Functional predictors  are constructed using a B-spline basis with \(M=15\) basis functions:
\begin{align*}
X_{i,1}(t) = U_{i}\sin(2\pi t), 
\quad
X_{i,2}(t) = V_{i}\sin(10\pi t) + \delta_i(t),
\end{align*}
where \(\delta_i(t)\sim \mathcal{N}(0,0.01^2)\). Scalar predictors \(\mathbf{Z}_i\in\mathbb{R}^{50}\) are generated in rank-deficient form as
$
\mathbf{Z}_i = \mathbf{A}
(
U_{i},
V_{i}
)^\top
+ \boldsymbol{\delta}_i,
$
with \(\mathbf{A}\in\mathbb{R}^{50\times 2}\) having i.i.d. \(\mathrm{Unif}(-1,1)\) entries  \(\boldsymbol{\delta}_i\) is drawn from standard bivariate normal.
The response is defined as
$
Y_i = 0.5\, t_{i,1} + 10\, t_{i,2} + \epsilon_i$,
where $\epsilon_i\sim\mathcal{N}(0,1),
$
so that the low-variance factor \(t_2\) dominates the response. Although the functional predictors enter nonlinearly, the properties under consideration still hold, as they are purely algebraic.

We extract \(L=10\) components under three regularization regimes--\emph{weak}, \emph{mixed}, and \emph{strong}--with \((\lambda_1,\lambda_2)=(0.1,0.1)\), \((0.1,10)\), and \((10,10)\), respectively, and report the maximum deviations from orthogonality between: residualized predictors and PLS directions, between PLS directions, and between PLS scores:
\begin{equation*}
\max_{1 \leq k < l \leq L}
\sqrt{
\sum_{i=1}^n
\langle 
\widehat{\xi}^{[k]},
\widetilde{W}_i^{[l]}  \rangle_\mathbb{H}^2
}
\quad
%%%%%%%%%%%%%%
\max_{1 \leq k < l \leq L}|\langle \widehat{\xi}^{[k]}, \widehat{\xi}^{[l]} \rangle_{\mathbb{H},\Lambda}|
%%%%%%%%%%
\quad
\max_{1 \leq k < l \leq L}
\sqrt{
\sum_{i=1}^n
\mathrm{Cor}(\widehat{\rho}_i^{[k]}, \widehat{\rho}_i^{[l]})^2
}
\end{equation*}  
Table~\ref{tab:geometric_validation} shows that all metrics remained negligible (on the order of $10^{-11}$ to $10^{-16}$) across scenarios. This demonstrates that the orthogonality properties hold robustly across regularization levels, confirming the algorithm's numerical stability under varying hyperparameter settings.

\begin{table}[ht]
\centering
\caption{Verification of orthogonality: residualized predictors vs. extracted directions (Lemma~\ref{lemma:annihilation}), PLS directions (Proposition~\ref{proposition: modified orthnormality of PLS components}), and PLS scores (Proposition~\ref{proposition: orthnormality of PLS scores}) under three regularization scenarios. Entries show mean maximum deviation from zero over 100 Monte Carlo replications (SEs in parentheses).}
\label{tab:geometric_validation}
\begin{tabular}{lccc}
\toprule
 &
\begin{tabular}[c]{@{}c@{}}
$\displaystyle
\max_{\substack{1 \le k < l \le L}}
\sqrt{\sum_{i=1}^n \langle 
\widehat{\xi}^{[k]},
\widetilde{W}_i^{[l]}
\rangle_\mathbb{H}^2}
$
\end{tabular} &
\begin{tabular}[c]{@{}c@{}}
$\displaystyle
\max_{\substack{1 \le k < l \le L}}
\left|\langle \widehat{\xi}^{[k]}, \widehat{\xi}^{[l]} \rangle_{\mathbb{H},\Lambda}\right|
$  
\end{tabular} &
\begin{tabular}[c]{@{}c@{}}
$\displaystyle
\max_{\substack{1 \le k < l \le L}}
\sqrt{\sum_{i=1}^n \mathrm{Cor}(\widehat{\rho}_i^{[k]}, \widehat{\rho}_i^{[l]})^2}
$  
\end{tabular}
\\
&($\times 10^{-15}$)&($\times 10^{-11}$)&($\times 10^{-16}$)
\\
\midrule
Weak   & $2.20\ (0.64)$ & $0.10\ (0.03)$ & $8.12\ (3.66)$ \\
Mixed  & $2.33\ (0.87)$ & $9.76\ (3.72)$ & $8.49\ (3.87)$ \\
Strong & $2.33\ (0.73)$ & $9.59\ (3.86)$ & $8.22\ (4.19)$ \\
\bottomrule
\end{tabular}
\end{table}

Next, we validate the ability of the NIPALS algorithm  to extract latent components in decreasing order of correlation with the response variable. We analyze the absolute sample correlations between the responses $Y_i$ and the estimated PLS scores $\widehat{ \rho}_i^{[l]}$ under the mixed regularization scenario.

\begin{figure}[t!]
    \centering
    \includegraphics[width=0.55\linewidth]{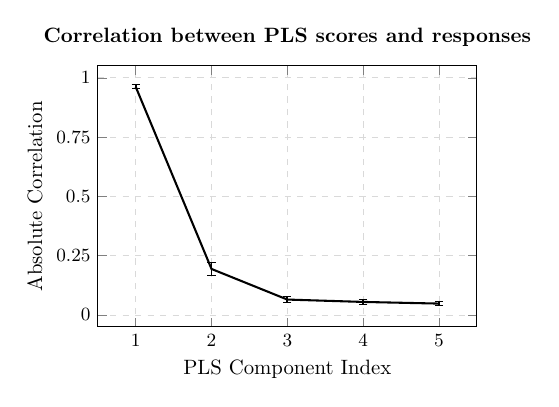}
     \caption{Mean absolute sample correlations between the response vector $Y_1, \ldots, Y_n$ and the latent PLS scores $\widehat{\rho}_1^{[l]}, \ldots, \widehat{\rho}_n^{[l]}$ for $l = 1, \ldots, 5$, averaged over 100 Monte Carlo replications. Error bars indicate one standard deviation.}
    \label{fig:pls_correlation}
    \label{fig:correlation_response}
\end{figure}
 As illustrated in Figure \ref{fig:correlation_response}, the first component exhibits the highest correlation with the response, effectively capturing the primary predictive signal. The second component accounts for a smaller but non-negligible portion of the remaining association, followed by  monotonic decay in correlation for all subsequent components.

\subsection{Regression Coefficient Estimation}\label{section:simul:beta}
To assess the sensitivity of coefficient estimation to regularization, we defined two cubic B-spline coefficient functions: $\beta_1(t) = 2t \sin(3\pi t)$, a modulated wave, and $\beta_2(t) = 2 \exp\left(-10(t-0.5)^2\right)$, a unimodal Gaussian bump, with scalar coefficients $\boldsymbol{\beta} = (1.5, -1.0)^\top$. Functional predictors were generated using 20 B-spline basis functions. The first functional predictor's coefficients were independent standard normal, while the second predictor's coefficients were a mixture of independent zero-mean Gaussian noise (weight 0.6) and the first predictor's coefficients (weight 0.4). The scalar predictor was drawn from a bivariate Gaussian distribution with mean equal to the first two coefficients of the first functional predictor and covariance matrix $0.5^2 \mathbf{I}_2$.
Responses were generated from a PFLM with additive zero-mean Gaussian noise, with the standard deviation set to 5\% of the signal's standard deviation.

We evaluated performance across 100 replications for sample sizes $n \in \{200, 1000, 3000\}$. A sensitivity analysis was conducted over a fixed regularization grid $(\lambda_1, \lambda_2) \in \{0, 0.001, 0.1\} \times \{0, 0.001, 0.1\}$, with model complexity fixed at $L = 10$. Performance was measured using relative $\mathbb{L}^2$ estimation error for functional coefficients and relative $\ell_2$ estimation error for scalar coefficients, defined as
$
\|\hat{\beta}_k - \beta_k\|_{\mathbb{L}^2} / \|\beta_k\|_{\mathbb{L}^2}
$
and
$
\|\hat{\boldsymbol{\beta}} - \boldsymbol{\beta}\|_2 / \|\boldsymbol{\beta}\|_2,
$
respectively.  
Results in Figure~\ref{fig:beta_sensitivity} indicate that regularization plays an important role in estimation accuracy. The results also suggest an interaction between the two regularization parameters, which is expected since, in the estimation procedure described in Theorem~\ref{proposition:linear_regul}, all three components are influenced by both $\lambda$ values through normalization by 
$
q = \sum_{j=1}^K \mathbf{u}_j^\top (\mathbf{B} + \lambda_j \mathbf{B}^{\prime\prime})^{-1} \mathbf{u}_j + \mathbf{v}^\top \mathbf{v}.
$

\begin{figure}[H] % requires \usepackage{float}, forces figure to appear here
    \centering
    \includegraphics[width=\textwidth]{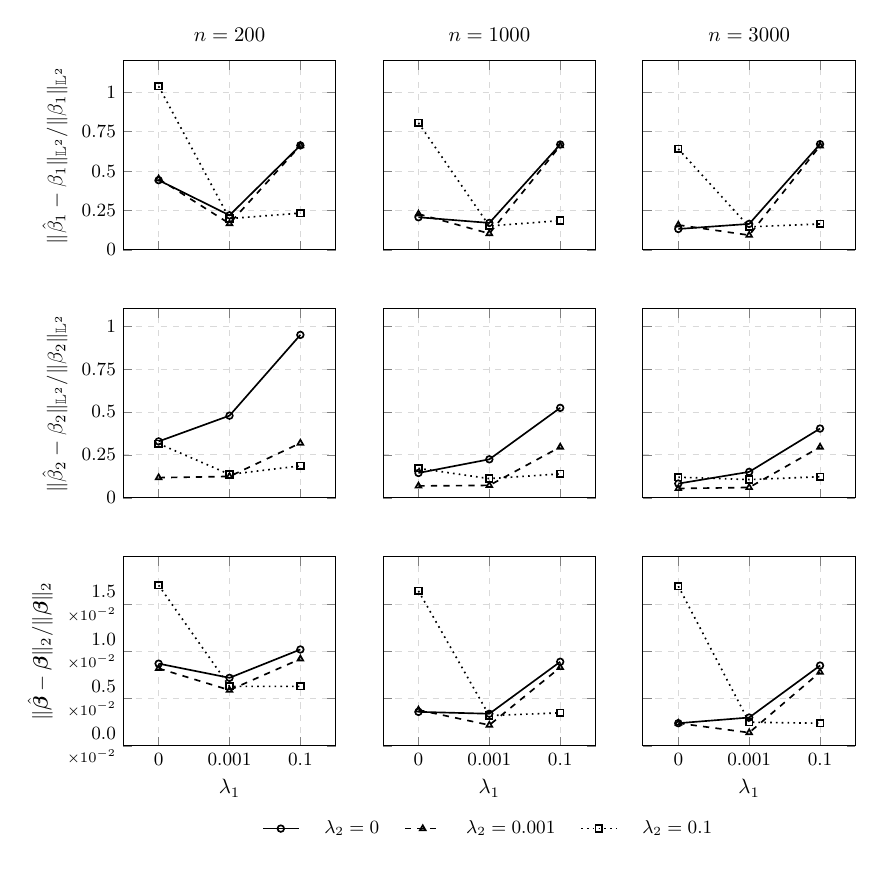}
\vspace{-1.5em}
    \caption{Sensitivity of relative estimation errors to the roughness penalties $\lambda_1$ and $\lambda_2$. Top and middle rows show the relative $\mathbb{L}^2$ error for functional coefficients $\beta_1(t)$ and $\beta_2(t)$, while the bottom row shows the relative $\ell_2$ error for the scalar coefficient $\boldsymbol{\beta}$. Columns correspond to sample sizes $n = 200, 1000, 3000$. The $x$-axis represents $\lambda_1$, with line types indicating $\lambda_2 \in \{0, 0.001, 0.1\}$. Reported values are mean across 100 Monte Carlo replications.}
    \label{fig:beta_sensitivity}
\end{figure}

\subsection{Predictive Power and Parsimony on Synthetic Data}\label{section:simul_comparison}
We evaluate the efficacy of our method in handling complex dependency structures through two synthetic experiments. We benchmark the predictive performance of Hybrid PLS against the PFLR method of \cite{xu_estimation_2020}. This method fits FPCA separately for each functional predictor and combines the resulting scores with the scalar predictors for ordinary least squares  regression. While this approach has been shown to be minimax optimal for prediction accuracy under the assumption of no multi-modal correlation, it is not explicitly designed to accommodate the complex dependency structures considered in this study. 
To ensure a fair comparison, we extend their approach by also applying PCA to the scalar predictors; we refer to this modified benchmark as principal component   regression (PCR) throughout this section.

\paragraph{Scenario 1: 
Orthogonal Nuisance Variance.}
To assess Hybrid PLS's ability to identify   subtle predictive structure in the presence of dominant but non-predictive variation, we construct a simulation in which the primary source of predictor variance is mathematically orthogonal to the response.
We defined two latent factors for each subject \( i = 1, \dots, 400 \): a predictive signal \( U_i \sim \mathcal{N}(0,1) \) and a nuisance factor \( V_i \). To remove any predictive information, \( V_i \) was constructed as the residual from regressing a draw from \( \mathcal{N}(0,1) \) onto the response \( Y_i \), and then scaled by a factor of 5 to dominate total variance.
The response was defined as
$
Y_i = 2 U_i + \epsilon_{i},
$
where \( \epsilon_{i} \) is centered Gaussian noise with standard deviation equal to 5\% of the sample standard deviation of \( 2 U_i \).
Functional predictors were generated as
\begin{align*}
X_{i1}(t) &= V_i \sin(4\pi t) + U_i \sin(2\pi t) + e_{i1}(t), \\
X_{i2}(t) &= V_i \cos(4\pi t) + U_i \cos(2\pi t) + e_{i2}(t),
\end{align*}
with \( e_{ik}(t) \sim \mathcal{N}(0, 0.1^2) \).
Scalar predictors \( \mathbf{Z}_i \in \mathbb{R}^5 \) were generated as
$
Z_{ij}
=
V_i + \delta_{ij}
$
for
$j = 1,\dots,4$
and
$
Z_{i5} = U_i + \delta_{ij}
$ where
$\delta_{ij} \sim \mathcal{N}(0, 0.1^2)$.
We split the data into 50\% training and 50\% test sets and evaluated test-set RMSE. For PCR, we additionally computed inter-modality correlations among the scores forming each combined component.

Figure~\ref{fig:simul_1}-(a) shows that,
under Scenario 1,
Hybrid PLS outperforms PCR in identifying the predictive component at the first step. PCR incurs a large first-component error (RMSE $\approx 0.66$), reflecting its emphasis on a high-variance nuisance factor orthogonal to the response. In contrast, Hybrid PLS leverages response covariance to suppress this variation, yielding a substantially lower error (RMSE $\approx 0.25$). Performance converges after the second component, consistent with a single signal latent factor in the data-generating process.
In addition,
Figure~\ref{fig:simul_1}-(b) shows that, under Scenario 1,
PCR   exhibits nearly perfect  cross-modality PC score correlations for the first two components.   This indicates that the first two components primarily capture $V_i$ and $U_i$,  resulting in redundant representations and inflated variance.

\begin{figure}[b!]
     \centering
     \includegraphics[width=0.95\linewidth]{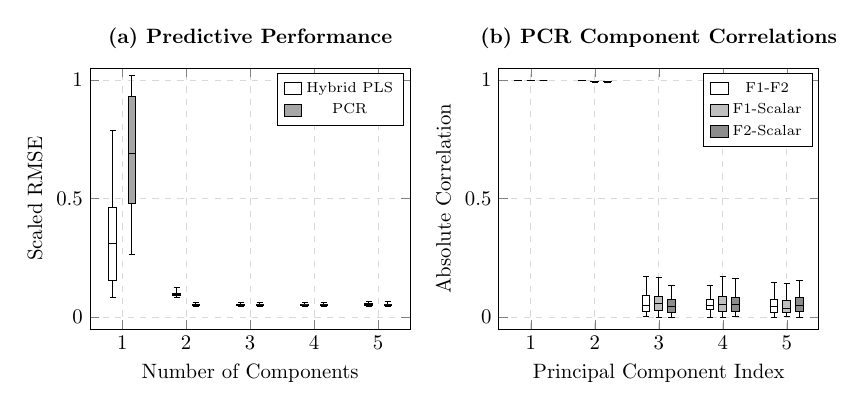}
 \caption{Scenario 1:   a predictive latent factor and a high-variance nuisance component.
    (a) Predictive performance (scaled test RMSE) of Hybrid PLS vs. PCR. 
    For PCR, the x-axis indicates $L$ components per source (two functional, one scalar), yielding $3L$ total predictors.
    (b) Absolute correlations between PC  scores from the three predictor sources. Boxplots show distributions across 100 replications.   }
     \label{fig:simul_1}
\end{figure}

\paragraph{Scenario 2: Function-Driven Cross-modality Correlation.} We next consider a setting where cross-modality correlation is induced by one functional predictor. 
Data were generated for \(n = 200\) subjects. Each functional predictor was represented using 20 B-spline basis functions. To induce realistic multicollinearity, the basis coefficients for the first functional predictor  were sampled independently from a standard normal distribution, while the coefficients for 
 the second functional predictor 
  were constructed as a weighted mixture of independent standard normal noise (weight 0.6) and the coefficients of the first functional predictor (weight 0.4). Cross-modal dependence was introduced by defining the scalar vector \(\mathbf{Z}\) using the first six basis coefficients of the first functional predictor with added Gaussian noise \(N(0, 0.5^2)\).
The functional regression coefficients were specified as \(\beta_1(t) = 2t \sin(3\pi t)\) and \(\beta_2(t) = 2 \exp\{-10(t - 0.5)^2\}\), both expressed in the same B-spline basis. The scalar regression coefficients were set to \(\boldsymbol{\beta} = (0.3, -0.2, 0.3, -0.2, 0.3, -0.2)\). The noiseless response was generated from a linear model combining functional and scalar effects, and the observed response \(Y\) was obtained by adding zero-mean Gaussian noise with standard deviation equal to 5\% of that of the noiseless response. We split the data into 50\% training and 50\% test sets and evaluated test-set RMSE. For PCR, we additionally computed inter-modality correlations among the scores forming each combined component.

 Figure \ref{fig:simul_1}-(a) shows that,
 under Scenario 2, 
 Hybrid PLS demonstrates superior parsimony, predictive power, and stability. In the first component, it achieves a significantly lower mean error (RMSE $\approx 0.27$ vs. $\approx 0.74$) and a substantially smaller standard error compared to PCR. While the predictive performance of the two methods converges after the second component, Hybrid PLS exhibits consistently lower variance across all number of components. Furthermore, Figure \ref{fig:simul_1}-(b) reveals  strong correlations   between the functional and scalar PC scores, confirming that PCR extracts redundant information across modalities.
\begin{figure}[b!]
     \centering
     \includegraphics[width=0.95\linewidth]{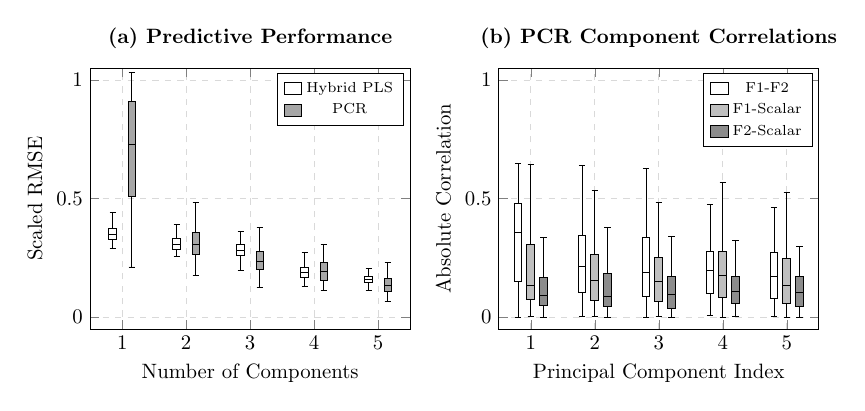}
\caption{Scenario 2: one functional latent factor.
  (a) Predictive performance (scaled test RMSE) of Hybrid PLS vs. PCR. 
    For PCR, the x-axis indicates $L$ components per source (two functional, one scalar), yielding $3L$ total predictors.
    (b) Absolute correlations between PC  scores from the three predictor sources. Boxplots show distributions across 100 replications. }
     \label{fig:simul_2}
\end{figure}

\subsection{Application to Renal Study Data}\label{section:real_data}
We evaluate our method on the Emory University renal study data \citep{changBayesianLatentClass2020}, which comprises 216 kidneys evaluated via diuretic renal scans. Following MAG3 injection, two functional predictors were recorded for each subject: (i) baseline renogram (time–activity) curves measuring MAG3 photon counts at 59 time points over 24 minutes, and (ii) post-furosemide curves measuring MAG3 photon counts at 40 time points during a subsequent 20-minute scan after furosemide administration. Several interpretable features of renogram curves are indicative of the likelihood of kidney obstruction. Typically, an unobstructed kidney is characterized by rapid MAG3 uptake followed by prompt excretion, suggesting no blockage (see dashed lines in the first two panels of Figure~\ref{fig:EDA}). In contrast, an obstructed kidney often exhibits prolonged MAG3 accumulation in the baseline curve, followed by poor excretion into the bladder throughout the post-furosemide renogram (see solid lines in the first two panels of Figure~\ref{fig:EDA}). However, in practice, substantial kidney-to-kidney variability exists in renogram curves, and many exhibit patterns that are not clear-cut, such as equivocal cases (see dotted lines in Figure~\ref{fig:EDA}).

To account for inter-subject variability in tracer uptake efficiency, both functional curves were smoothed using a B-spline basis with 20 basis functions and normalized by the peak intensity of the subject's baseline renogram. Also the time domain was normalized into $t \in [0,1]$. In addition to the renogram curves, the dataset includes fifteen scalar covariates and three response variables. The scalar covariates consist of subject age and fourteen summary statistics derived from the baseline and post-furosemide renogram curves, while the response variables are diagnosis scores provided by three nuclear medicine experts. See Table \ref{tab:renogram_variables} for details. Because the summary statistics are direct transformations of the functional predictors, the data exhibit substantial cross-modality correlation. Although conventional analyses often rely exclusively on these scalar covariates, we directly incorporate the functional predictors to leverage latent information in the renogram curves that is not captured by the scalar summaries. The response variable $Y$ is defined as the average obstruction rating across the three experts, min-max scaled to the unit interval $[0,1]$, with larger values indicating a greater degree of kidney obstruction. Specifically, values closer to 1 correspond to higher obstruction severity, values near 0.5 indicate equivocal cases, and values closer to 0 indicate a high likelihood of non-obstruction.

\begin{table}[htbp]
\centering
\caption{Eighteen renogram variables (4 clinical and 14 pharmacokinetic variables).}
\label{tab:renogram_variables}
\begin{tabular}{clp{9cm}}
\toprule
Index & Variable & Description \\
\midrule
1  & Age & Range: 18 to 87 years \\

2--4 & Expert's score & Ranges from $-1.0$ to $1.0$, with higher scores (from three experts) indicating a higher likelihood of obstruction \\

5  & Cortical AUC-d1 & Area under the first derivative of the cortical renogram \\

6  & BL AUC-d1 & Area under the first derivative of the baseline (BL) renogram \\

7  & Cortical mv & Minimum velocity (mv) of the cortical renogram \\

8  & BL mv & Minimum velocity (mv) of the BL renogram \\

9  & Cortical tminv & Time to minimum velocity (tminv) of the cortical renogram \\

10 & BL tminv & Time to minimum velocity (tminv) of the BL renogram \\

11 & Cortical tmax & Time to maximum (tmax) of the cortical renogram \\

12 & BL tmax & Time to maximum (tmax) of the BL renogram \\

13 & PF AUC & Area under the post-furosemide (PF) renogram \\

14 & Pelvis AUC & Area under the renogram of the pelvis region \\

15 & PF max & Maximum (max) of the PF renogram \\

16 & Pelvis max & Maximum (max) of the pelvis renogram \\

17 & lastPF/maxBL & Ratio of PF renogram at the last time point to BL maximum \\

18 & firstPF/maxBL & Ratio of PF renogram at the first time point to BL maximum \\
\bottomrule
\end{tabular}

\vspace{0.5em}
\footnotesize
\textbf{Abbreviations:} BL, baseline; PF, post-furosemide.
\end{table}

To evaluate predictive performance and stability, we conducted 100 replications using a subsampling strategy. In each replication, a random 40\% subset of subjects (yielding approximately 86 kidneys) was selected and partitioned into 70\% training and 30\% testing sets. We report the range-normalized test RMSE. We benchmarked our approach against Penalized Functional Regression (PFR; \citealp{goldsmithPenalizedFunctionalRegression2011}), implemented via the \texttt{refund} package (details in Table \ref{tab:pflm_summary}), which utilizes all predictors, as well as a OLS model using only the scalar covariates.

As illustrated in Figure \ref{fig:real_data_prediction_subplot}-(a), Hybrid PLS captures the essential regression signal within its first component, outperforming both PFR and OLS
in mean and variance while matching the predictive power of PCR. Hybrid PLS achieves this with substantially greater parsimony: comparable performance is obtained using only $L$ Hybrid PLS components, whereas PCR requires  $3L$ principal component scores---$L$ extracted from each source (two functional and one scalar)---to reach a similar level of predictive accuracy.

Furthermore, 
Figure \ref{fig:real_data_prediction_subplot}-(b) demonstrates that
Hybrid PLS maintains stable RMSE and standard errors as the number of components increases, whereas PCR exhibits rising errors and instability indicative of overfitting.
These findings suggest that Hybrid PLS effectively isolates variation correlated with the obstruction diagnosis, whereas PCR is prone to capturing non-predictive noise.
The significant improvement over OLS confirms that our method successfully extracts latent information from the renogram curves that is not encoded in the scalar summaries. In contrast, PFR yields negligible improvement in prediction accuracy over OLS while exhibiting higher variance. Furthermore, it performs worse than both PCR and Hybrid PLS. The observation that PFR displays greater instability than OLS suggests that the PFR's reliance on regularization without explicit dimension reduction was ineffective in managing cross-modal correlations in this context, ultimately hindering predictive performance.

\begin{figure}[b!]
     \centering
     \includegraphics[width=0.95\linewidth]{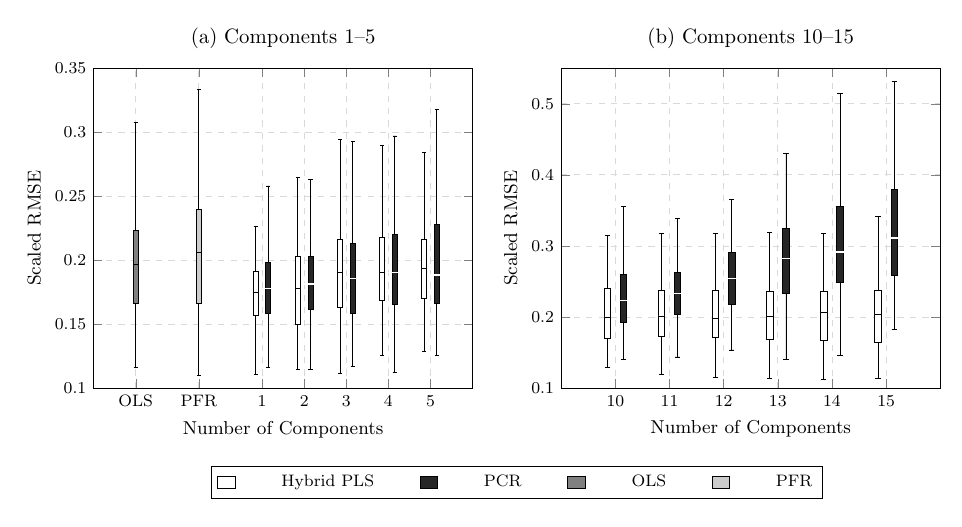}
\caption{Predictive performance on the kidney dataset. Scaled test RMSE is compared for Hybrid PLS, PCR, Scalar OLS, and PFR across 100 subsampling replications. Panel (a) shows results for the baseline methods (Scalar OLS and PFR) alongside the performance of Hybrid PLS and PCR using 1-5 components (for PCR, $L$ components are extracted per source---two functional and one scalar---yielding a total of $3L$ predictors). Panel (b) presents the performance of Hybrid PLS and PCR using 10-15 components. Boxplots summarize the distribution of scaled test RMSE across the 100 replications.
}
     \label{fig:real_data_prediction_subplot}
\end{figure}

Figure~\ref{fig:kidney_cor} shows strong inter-source correlations among the PCR principal component scores, particularly for the first component. Figure~\ref{fig:kidney_weights} further demonstrates that the first components extracted by Hybrid PLS and PCR capture clinically meaningful features of the renogram curves associated with the obstruction mechanism. Specifically, consistent with the typical pattern described above---where obstructed kidneys exhibit slower initial MAG3 accumulation than non-obstructed kidneys, followed by a steady increase throughout the baseline renogram, whereas non-obstructed kidneys show early MAG3 drainage---the first component displays negative values in the early portion of the curve and positive values in the later portion. We also observe that, for the functional predictors, PLS extracts a comparatively simpler (smoother) component structure. For the scalar predictors, the relative variable weights differ substantially between methods: notably, variable 14 ranks second in importance under PLS but does not appear among the top five variables under PCR. Overall, PLS assigns smaller weights across components, likely reflecting its unified treatment of cross-modality correlations and built-in regularization, rather than being influenced by correlations across separately constructed principal components.

\begin{figure}
    \centering
    \includegraphics[width=0.65\linewidth]{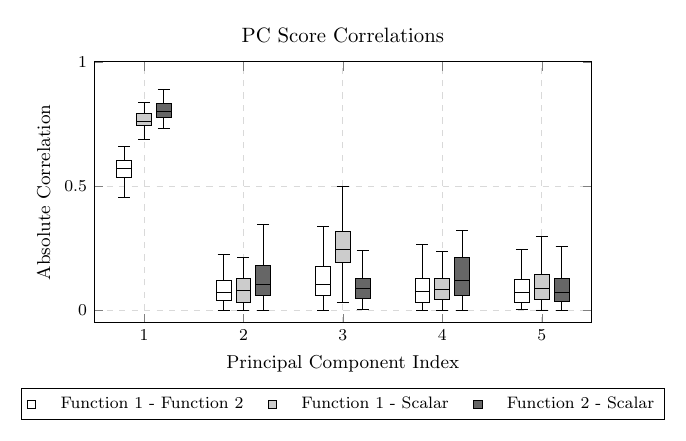}
\caption{Absolute correlations between PCR component scores from the three predictor sources on the kidney dataset. Boxplots show distributions across 100 subsampling replications by component index.}
    \label{fig:kidney_cor}
\end{figure}
%%%%%%%%%%%%%%%%%%%%%%%%%%%%%%%
\begin{figure}[t!]
\centering
\includegraphics[width=0.95\textwidth]{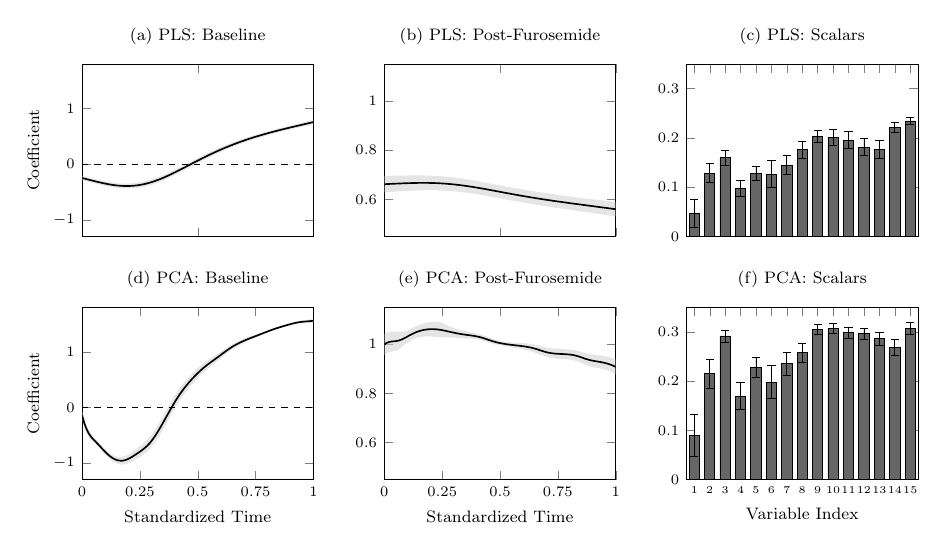}
\caption[Comparison of first component directions for Hybrid PLS and PCR]{%
\textbf{Comparison of the first latent component directions for Hybrid PLS (top row) and PCR (bottom row).}
The panels are organized by predictor source:
\textbf{(a, d)} estimated functional weights for the baseline renogram;
\textbf{(b, e)} estimated functional weights for the post-furosemide renogram; and
\textbf{(c, f)} coefficient loadings for the scalar covariates.
In the functional plots, solid curves denote the pointwise mean across 100 subsampling replications, with shaded bands representing the $\text{mean} \pm 1 \text{ SD}$ variability. For the scalar covariates, bars indicate the mean absolute loading and error bars span $\pm 1 \text{ SD}$.}
\label{fig:kidney_weights}
\end{figure}
\section{Discussion}\label{discussion}
In this article, we proposed a Hybrid Partial Least Squares (Hybrid PLS) framework to address the challenge of simultaneously modeling multiple functional and scalar predictors. By constructing a unified hybrid Hilbert space, we rigorously extended the classical NIPALS algorithm to jointly manage the complex dependency structures---particularly cross-modality correlations---that are often overlooked by conventional two-stage approaches. We established the theoretical well-posedness of the population-level problem and proved the consistency of the proposed estimator. From a computational perspective, we derived a closed-form solution for the regularized iterative steps, allowing the method to efficiently handle dense or irregular functional data while strictly enforcing geometric properties such as score orthogonality.
Our numerical studies and the application to the Emory University renal dataset highlight the practical advantages of this unified approach. 
For future research direction, deriving theoretical convergence rates and minimax optimality bounds would provide a deeper understanding of its statistical efficiency relative to other high-dimensional regression techniques.

\bibliography{bibliography}
%%%%%%%%%%%%%%%%%%%%%%%%%%%%%%%%%%%%%%%%%%%%%%
%% Example with single Appendix:            %%
%%%%%%%%%%%%%%%%%%%%%%%%%%%%%%%%%%%%%%%%%%%%%%
\begin{appendix}\label{appendix}
 
\section{Overview of Appendix}
This appendix provides the mathematical foundations and detailed proofs for the Hybrid PLS framework. The content is organized as follows:

\begin{itemize}
    \item Appendix \ref{section:proof:lemma:hybrid_hilbert_space} verifies that the product space $\mathbb{H}$ defined in \textbf{Section \ref{section:hybrid_hilbert}} satisfies the properties of a separable Hilbert space.

        \item Appendix \ref{section:proof:section:hybrid_hilbert} establishes the well-posedness of the population-level framework in \textbf{Section \ref{section:hybrid_hilbert}}, providing proofs for the compactness of the covariance operator (\textbf{Lemma \ref{lemma:cross_cov_functional}}), the population PLS eigenproblem (\textbf{Theorem \ref{theorem:population_PLS}}),
    the equivalence of residualization and orthogonality (\textbf{Theorem \ref{prop:residualization_equiv_eigen}}) and the mean-squared convergence of the estimator (\textbf{Theorem \ref{thm:hybrid_pls_convergence}}).
    
    \item Appendix \ref{section:proof:section:sub:iterative} provides the derivations for the iterative sample-level algorithm presented in \textbf{Section \ref{section:sub:iterative}}. This includes the eigenproblem formulation for the regularized PLS direction (\textbf{Theorem \ref{proposition:eigen_regul}}), its closed-form solution (\textbf{Theorem \ref{proposition:linear_regul}}), and the residualization updates (\textbf{Lemma \ref{proposition:closed_form_orthgonalization}}).
    
    \item Appendix \ref{section:proof:lemma:recursive} proves the recursive linear relationship between the predictors and the PLS directions established in \textbf{Lemma \ref{lemma:recursive}} (\textbf{Section \ref{section:sub:regression_coeff}}).

    \item Appendix \ref{section:appendix:geom} proves the geometric properties of the sample algorithm discussed in \textbf{Section \ref{section:sub:geom}}, specifically the orthonormality of the estimated directions (\textbf{Theorem \ref{proposition: modified orthnormality of PLS components}}) and the orthogonality of the PLS scores (\textbf{Theorem \ref{proposition: orthnormality of PLS scores}}).
\end{itemize}
	\section{Separability of the hybrid space}\label{section:proof:lemma:hybrid_hilbert_space}

		Both $\mathbb{L}^2[0,1]$ and $\mathbb{R}^p$ are Hilbert spaces, with their inner products, norms, and metrics corresponding to the terms defined in equation \eqref{eq: hybrid inner product}. 
        Let us define the metric of $\mathbb{H}$ as
         $d(h_1, h_2) = \Vert h_1 - h_2 \Vert_{\mathbb{H}}$.
        A finite Cartesian product of Hilbert spaces, equipped with an $\ell_2$ norm of the component metrics (which is our metric), is also a Hilbert space.
		Furthermore, both of 	 $\mathbb{L}^2[0,1]$ and $\mathbb{R}^p$ are separable.
		For $\mathbb{L}^2[0,1]$, the set of all polynomials with rational coefficients is a countable and dense subset.
		For $\mathbb{R}^p$,  the set of all vectors with rational components forms a countable and dense subset.
		A finite Cartesian product of separable spaces is also separable.
		Therefore $\mathbb{H} = \bigl( \mathbb{L}^2[0,1] \bigr)^K \times \mathbb{R}^p$ is separable.

	\section{Proof of Section \ref{section:hybrid_hilbert}}\label{section:proof:section:hybrid_hilbert}
	This section provides proof of Theorem \ref{theorem:population_PLS}, 
	demonstrating that the population-level version of the optimization problem solved at each step is well-defined, and 
	proofs of the intermediate lemmas that support the main theorem.
	\subsection{Proof of  Lemma \ref{lemma:cross_cov_functional} }\label{section:proof:lemma:cross_cov_functional}
	\begin{proof} We first show that
    the operator $\mathcal{C}_{YW}$ is a compact operator under condition \eqref{condition_for_compact}
		Any bounded linear functional from a Hilbert space to  $\mathbb{R}$  is a compact operator.
		This is because the image of any bounded set under such a functional is a bounded set in $\mathbb{R}$,
		and by Bolzano-Weierstrass Theorem, every bounded sequence of real numbers has a convergent subsequence.
		Therefore,  to show that $\mathcal{C}_{YW}$ is a compact operator, it suffices to show that 
		$\mathcal{C}_{YW}$ is a bounded linear functional.
		
		\medskip
		\noindent
		\textit{Linearity.}~	  
		The operator $\mathcal{C}_{YW}: \mathbb{H} \rightarrow \mathbb{R}$ is defined as 
		$\mathcal{C}_{YW} h = \langle \Sigma_{YW}, h \rangle_\mathbb{H}$.
		For any $h_1, h_2 \in \mathbb{H}$ and scalar $c \in \mathbb{R}$, the linearity of the inner product implies: \begin{align*} \mathcal{C}_{YW}(h_1 + h_2) &= \langle \Sigma_{YW}, h_1 + h_2 \rangle_\mathbb{H} = \langle \Sigma_{YW}, h_1 \rangle_\mathbb{H} + \langle \Sigma_{YW}, h_2 \rangle_\mathbb{H} = \mathcal{C}_{YW} h_1 + \mathcal{C}_{YW} h_2,
			\\ 
			\mathcal{C}_{YW}(c h) &= \langle \Sigma_{YW}, c h \rangle_\mathbb{H} = c \langle \Sigma_{YW}, h \rangle_\mathbb{H} = c \, \mathcal{C}_{YW} h. \end{align*}
		Thus, $\mathcal{C}_{YW}$ is a linear functional. 
		
		\medskip
		\noindent
		\textit{Boundedness.}~ 
		To show that $\mathcal{C}_{YW}$ is bounded, we need to find a finite constant $M$ such that $|\mathcal{C}_{YW} h| \leq M \|h\|_\mathbb{H}$ for all $h \in \mathbb{H}$. By the Cauchy-Schwarz inequality, we have: \begin{equation*} 
			|\mathcal{C}_{YW} h| = |\langle \Sigma_{YW}, h \rangle_\mathbb{H}| \leq \|\Sigma_{YW}\|_\mathbb{H} \|h\|_\mathbb{H}. 
		\end{equation*}
%%%%%%%%%%%%%%%%
Now, we must verify that $\|\Sigma_{YW}\|_\mathbb{H}$ is finite. 
Let $\mu([0,1])$ denote a Lebesgue measure of $[0,1]$, and let $T = \max_{k=1, \ldots, K} \, \mu([0,1])$.
The squared norm of $\Sigma_{YW}$   is given by:
\begin{equation*}
\|\Sigma_{YW}\|_\mathbb{H}^2 
= 
\langle \Sigma_{YW}, \Sigma_{YW} \rangle_\mathbb{H}
=
\sum_{k=1}^K
    \int_0^1 \mathbb{E} \left[   Y X_k\right](t)^2 \, dt
+ \sum_{r=1}^p \mathbb{E} \left[   Y Z_r\right]^2
			<
			K T Q_1 + p  Q_2 < \infty,
		\end{equation*} 
		where the last inequality uses the condition \eqref{condition_for_compact}.
		Let $M = \sqrt{K T Q_1 + p  Q_2 }$. Thus, we have shown that $|\mathcal{C}_{YW} h| \leq M \|h\|_\mathbb{H}$ for a finite constant $M$.

Next, we show that $\mathcal{U}$ is self-adjoint, positive-semidefinite, and compact, in turn.
		
		\medskip
		\noindent
		\textit{Self-adjoint.}~
		For any $h_1, h_2 \in \mathbb{H}$, we have
		\begin{align*}
			\langle 
			\mathcal{U} h_1, \, h_2 
			\rangle_\mathbb{H}  
			= 
			\bigl \langle 
			\langle h_1, \Sigma_{YW} \rangle_\mathbb{H} \Sigma_{YW},
			\,
			h_2
			\bigr \rangle_\mathbb{H}
			%%%%%%%%%%%%%%%%%%%%%%%%%%%%%%%%%%
			&= \langle h_1, \Sigma_{YW} \rangle_\mathbb{H} \langle \Sigma_{YW}, h_2 \rangle_\mathbb{H}
			%%%%%%%%%%%%%%%%%%%%%%%%%%%
			\\&= 
			\bigl \langle 
			h_1,
			\,
			\langle \Sigma_{YW}, h_2 \rangle_\mathbb{H} \Sigma_{YW} 
			\bigr \rangle_\mathbb{H}
			\\ 	&= \langle h_1, \mathcal{U} h_2 \rangle_\mathbb{H}.
		\end{align*}
		
		\medskip
		\noindent
		\textit{Positive-semidefinite.}~
		For every $h \in \mathbb{H}$, we have
		\begin{equation*}
			\langle 
			\mathcal{U} h,
			\,
			h
			\rangle_\mathbb{H}
			=
			\bigl \langle 
			\langle 
			\Sigma_{YW},\, h \rangle_{\mathbb{H}}  \Sigma_{YW}, \, h
			\bigr \rangle_\mathbb{H}
			= \langle \Sigma_{YW}, \, h \rangle_\mathbb{H}^2 \ge 0.
		\end{equation*}
		
		\medskip
		\noindent
		\textit{Compact.}~
		We have shown that $\mathcal{C}_{YW}$ is a compact operator  under the condition \eqref{condition_for_compact}. 
The composition of two operators is compact if either operator is compact
(Theorem 4.1.3 of \citealp{hsingTheoreticalFoundationsFunctional2015}). Therefore
		$\mathcal{U} := \mathcal{C}_{WY} \circ \mathcal{C}_{YW}$ is a compact operator.
		This completes the proof of Lemma \ref{lemma:cross_cov_functional}.
		
	\end{proof}

\begin{comment}
	\subsection{Proof of Lemma \ref{lemma:cross_cov_operator}}\label{section:proof:lemma:cross_cov_operator}
	\begin{proof}
		For any $h \in \mathbb{H}$ and $d \in \mathbb{R}$, we have:
		\begin{equation*}
			\langle
			\mathcal{C}_{YW} h, d \rangle 
			= 
			\mathbb{E}
			[
			\langle W_1, h \rangle_\mathbb{H} Y_1  
			] 	 \,  d
			%%%%%%%%%%%%%%%%%%
			=
			\mathbb{E} \langle   Y_1 W_1 d , \, h\rangle_{\mathbb{H}} 
			%%%%%%%%%%%%%%%%%%%%%%%%
			= \langle h, \, \mathbb{E}[Y_1 W_1 d   ]\rangle_{\mathbb{H}}
			=
			\langle h, \mathcal{C}_{WY} d \rangle_{\mathbb{H}}
		\end{equation*}
		Thus, we have $\langle \mathcal{C}_{YW} h, d \rangle =  \langle h, \mathcal{C}_{WY} d \rangle_{\mathbb{H}}$, which implies $\mathcal{C}_{WY} = \mathcal{C}_{YW}^\ast$.
		This completes the proof of Lemma \ref{lemma:cross_cov_operator}.
	\end{proof}
    \end{comment}

	\subsection{Proof of Theorem \ref{theorem:population_PLS}}\label{section:proof:theorem:population_PLS}
	\begin{proof}
Since $\mathcal{C}_{YW}$ is compact by Lemma \ref{lemma:cross_cov_functional}, it 
		has the singular value decomposition 
		$ 
			\mathcal{C}_{YW} = \sum \limits_{j=1}^\infty \kappa_j (f_{1j} \otimes f_{2j})
		$ (Theorem 4.3.1 of \citealp{hsingTheoreticalFoundationsFunctional2015}).
		Let $\Vert \cdot \Vert_{\mathrm{op}}$ denote an operator norm. We have
		\begin{equation*}
        \kappa_1^2 = 
			\Vert \mathcal{C}_{YW} \Vert_{\mathrm{op}} = \sup_{\substack{h \in \mathbb{H} \\ \Vert h \Vert_\mathbb{H}=1}} | \mathcal{C}_{YW} h |^2 
			= \sup_{\substack{h \in \mathbb{H} \\ \Vert h \Vert_\mathbb{H}=1}} | \mathbb{E}(\langle W, h \rangle_\mathbb{H} Y)|^2 
			= \sup_{\substack{h \in \mathbb{H} \\ \Vert h \Vert_\mathbb{H}=1}} \textnormal{Cov}^2(\langle W, h \rangle_\mathbb{H}, Y) 
			,
		\end{equation*}
		with maximum attained at $h=f_{11}$, which is an eigenfunction of  $\mathcal{C}_{YW}^* \, \circ \, \mathcal{C}_{YW} = \mathcal{C}_{WY} \, \circ \, \mathcal{C}_{YW} = \mathcal{U}$ corresponding to the largest eigenvalue $\kappa_1^2$(4.3.4 in \citealp{hsingTheoreticalFoundationsFunctional2015}).
		This completes the proof of Theorem \ref{theorem:population_PLS}.
	\end{proof}

 \subsection{Proof of Proposition \ref{prop:residualization_equiv_eigen}}\label{section:proof:prop:residualization_equiv_eigen}
  The proof procedes in the following steps.
 \begin{enumerate}
     \item Show that for every \(s\ge 1\) and every \(j\le s-1\), we have $
\mathbb{E}[\rho^{[j]}\rho^{[s]}]=0
$ for Algorithm \ref{alg:population_pls} (Appendix \ref{section:proof:orthogonal_score_pop}),
%%%%
\item Show that for every \(s\ge 1\) and every \(j\le s-1\), we have
$
\mathbb{E}[Y^{[s]}\rho^{[j]}]=0
$
for Algorithm \ref{alg:population_pls} (Appendix \ref{section:proof:orthogonal_score_residual_pop}),
\item Show the equivalence between Algorithm \ref{alg:population_pls} and optimization problem \eqref{const_eigen_pop} (Appendix \ref{equiv_problem}).
\end{enumerate}

\subsubsection{Orthogonality of the scores in Algorithm \ref{alg:population_pls}}\label{section:proof:orthogonal_score_pop}
\begin{proof}
Let
$
\Sigma_{W^{[l]}}
$ be defined as in Section \ref{section:operators}.
Since 
$
\rho^{[l]} = \langle W^{[l]}, \xi^{[l]} \rangle_{\mathbb{H}},
$
we can rewrite the numerator and denominator of $\delta^{[l]}$ as
\[
\mathbb{E}[ \rho^{[l]}  \,  W^{[l]} ] 
= 
\mathbb{E}[
\langle W^{[l]}, \xi^{[l]} \rangle_{\mathbb{H}} W^{[l]} ]
= \Sigma_{W^{[l]}} \, \xi^{[l]},
\] 
and
\[
\mathbb{E}[ \rho^{[l]}  \,  \rho^{[l]} ] 
= 
\mathbb{E}[
\langle W^{[l]}, \xi^{[l]} \rangle_{\mathbb{H}} \langle W^{[l]}, \xi^{[l]} \rangle_{\mathbb{H}} ]
= 
\langle \Sigma_{W^{[l]}} \, \xi^{[l]}, \xi^{[l]} \rangle_{\mathbb{H}},
\]
respectively.
Hence, at the population level, the residualization step can be written as
\[
W^{[l+1]} = W^{[l]} - 
\langle W^{[l]}, \xi^{[l]} \rangle_{\mathbb{H}}
\frac{ 
 1
}
     { \langle \xi^{[l]}, \Sigma_{W^{[l]}} \, \xi^{[l]} \rangle_{\mathbb{H}} }
     \Sigma_{W^{[l]}} \, \xi^{[l]} .
\]

For $l_1 < l_2$:
\[
\mathbb{E}[\rho^{[l_1]} \, \rho^{[l_2]}]
=
\mathbb{E}\big[ \rho^{[l_1]} \,
                 \langle W^{[l_2]}, \xi^{[l_2]} \rangle_{\mathbb{H}} \big]
=
\mathbb{E}\big[ 
                 \langle \rho^{[l_1]} \, W^{[l_2]}, \xi^{[l_2]} \rangle_{\mathbb{H}} \big]
=
\langle 
\mathbb{E}\big[ \rho^{[l_1]} \,W^{[l_2]}] \,
, \,
\xi^{[l_2]} \rangle_{\mathbb{H}}
\]
Therefore,
to show that $\mathbb{E}[\rho^{[l_1]} \, \rho^{[l_2]}]$,
it suffices to show $\mathbb{E}\big[ \rho^{[l_1]} \,W^{[l_2]}] = 0 \in \mathbb{H} $.

We can express
$\mathbb{E}\big[ \rho^{[l_1]} \,W^{[l_2]}]$ as
\begin{align*}
\mathbb{E}&\big[ \rho^{[l_1]} \,W^{[l_2]}]
%%%%%%%%%%%%%%%%%%%%%%%%%%%%%%%%%
\\&=
\mathbb{E}\big[ \rho^{[l_1]} \,W^{[l_2-1]}]
-
\mathbb{E}[
\rho^{[l_1]}
\frac{ 
\langle W^{[l_2-1]}, \xi^{[l_2-1]} \rangle_{\mathbb{H}}
}
     { \langle \xi^{[l_2-1]}, \Sigma_{W^{[l_2-1]}} \, \xi^{[l_2-1]} \rangle_{\mathbb{H}} }
     \Sigma_{W^{[l_2-1]}} \, \xi^{[l_2-1]}
]
%%%%%%%%%%%%%%%%%%%%%%%%%%%%%%%%%
\\&=
\mathbb{E}\big[ \rho^{[l_1]} \,W^{[l_2-1]}]
-
\mathbb{E}[
\rho^{[l_1]}
\langle W^{[l_2-1]}, \xi^{[l_2-1]} \rangle_{\mathbb{H}}
]
\frac{ 
 1
}{ 
\langle
    \xi^{[l_2-1]}, 
    \Sigma_{W^{[l_2-1]}} \, \xi^{[l_2-1]}
\rangle_{\mathbb{H}} }
     \Sigma_{W^{[l_2-1]}} \, \xi^{[l_2-1]}
%%%%%%%%%%%%%%%%%%%%%%%%%%%%%%%%%
\\&=
\mathbb{E}\big[ \rho^{[l_1]} \,W^{[l_2-1]}]
-
\langle 
\mathbb{E}[
\rho^{[l_1]} W^{[l_2-1]}], \xi^{[l_2-1]} \rangle_{\mathbb{H}}
\frac{ 
 1
}{ 
\langle
    \xi^{[l_2-1]}, 
    \Sigma_{W^{[l_2-1]}} \, \xi^{[l_2-1]}
\rangle_{\mathbb{H}} }
     \Sigma_{W^{[l_2-1]}} \, \xi^{[l_2-1]}
%%%%%%%%%%%%%%%%%%%%%%%%%%%%%%%%%
			\\	&=	H^{[l_2-1]} \bigl( 
            \mathbb{E}[\rho^{[l_1]  }
			W^{[l_2-1]}] \bigr)~(\text{say}).
		\end{align*}
Here, the deterministic operator $H^{[l_2-1]}: \mathbb{H} \mapsto \mathbb{H}$ maps the zero element of $\mathbb{H}$ to itself. Using this relationship, we have
\begin{align*}
\mathbb{E}\big[ \rho^{[l_1]} \,W^{[l_2]}]  
			%%%%%%%%%%%%%%%%
&=
H^{[l_2-1]} \bigl( \mathbb{E}\big[ \rho^{[l_1]} \,W^{[l_2-1]}] \bigr)
%%%%%%%%%%%%%%%%
\\&=
H^{[l_2-1]} 
\circ 
H^{[l_2-2]}
\bigl(
\mathbb{E}\big[ \rho^{[l_1]} \,W^{[l_2-2]}]
\bigr)
			%%%%%%%%%%%%%%%%
			\\&=
			H^{[l_2-1]} 
			\circ 
			H^{[l_2-2]}
			\circ 
			\ldots
			\circ 
			H^{[l_1]}
			\bigl( \mathbb{E}\big[ \rho^{[l_1]} \,W^{[l_1]}] \bigr).
		\end{align*}
Therefore,
to  show that $\mathbb{E}[\rho^{[l_1]} \, \rho^{[l_2]}]$,
it suffices to show $H^{[l_1]}
			\bigl( \mathbb{E}\big[ \rho^{[l_1]} \,W^{[l_1]}] \bigr) = 0 \in \mathbb{H}$. Note that
we have
\begin{equation*}
    \mathbb{E}[\rho^{[l_1]} W^{[l_1]}] = \mathbb{E}[\langle W^{[l_1]}, \xi^{[l_1]} \rangle_{\mathbb{H}} \, W^{[l_1]}] = \Sigma_{W^{[l_1]}} \, \xi^{[l_1]}.
\end{equation*}
Leveraging this, we have
\begin{align*}
H^{[l_1]}&
			\bigl( \mathbb{E}\big[ \rho^{[l_1]} \,W^{[l_1]}] \bigr)
 \\           &=
\mathbb{E}\big[ \rho^{[l_1]} \,W^{[l_1]}]
-
\langle 
\mathbb{E}[
\rho^{[l_1]} W^{[l_1]}], \xi^{[l_1]} \rangle_{\mathbb{H}}
\frac{ 
 1
}{ 
\langle
    \xi^{[l_1]}, 
    \Sigma_{W^{[l_1]}} \, \xi^{[l_1]}
\rangle_{\mathbb{H}} }
     \Sigma_{W^{[l_1]}} \, \xi^{[l_1]}
 \\           &=
\Sigma_{W^{[l_1]}} \, \xi^{[l_1]}
-
\langle 
\Sigma_{W^{[l_1]}} \, \xi^{[l_1]}, \xi^{[l_1]} \rangle_{\mathbb{H}}
\frac{ 
 1
}{ 
\langle
    \xi^{[l_1]}, 
    \Sigma_{W^{[l_1]}} \, \xi^{[l_1]}
\rangle_{\mathbb{H}} }
     \Sigma_{W^{[l_1]}} \, \xi^{[l_1]}
\\           &=
\Sigma_{W^{[l_1]}} \, \xi^{[l_1]}
-
\Sigma_{W^{[l_1]}} \, \xi^{[l_1]}
=0 \in \mathbb{H}.
\end{align*}
This completes the proof of
orthogonality of the scores in Algorithm \ref{alg:population_pls}.

\end{proof}

\subsubsection{Orthogonality of the scores and response residuals in Algorithm \ref{alg:population_pls}}\label{section:proof:orthogonal_score_residual_pop}
\begin{proof}
Fix \(s\ge 2\) and \(j\le s-1\). Recall the residualization for the response:
\[
Y^{[s]} \;=\; Y - \sum_{k=1}^{s-1} \nu^{[k]}\,\rho^{[k]},
\qquad
\nu^{[k]} \;=\; \frac{\mathbb{E}[Y^{[k]}\rho^{[k]}]}{\mathbb{E}[(\rho^{[k]})^2]}.
\]
Compute the covariance of interest by expanding the residual:
\begin{align*}
\mathbb{E}\big[ Y^{[s]}\,\rho^{[j]} \big]
= \mathbb{E}\Big[ \Big( Y - \sum_{k=1}^{s-1} \nu^{[k]}\rho^{[k]}\Big)\, \rho^{[j]} \Big] 
= \mathbb{E}[Y\,\rho^{[j]}] \;-\; \sum_{k=1}^{s-1} \nu^{[k]}\,\mathbb{E}[\rho^{[k]}\rho^{[j]}].
\end{align*}

By the score orthogonality already established in Appendix \ref{section:proof:orthogonal_score_pop},
we have \(\mathbb{E}[\rho^{[k]}\rho^{[j]}]=0\) for every \(k\neq j\). Thus only the \(k=j\) term survives in the sum and we obtain
\[
\mathbb{E}\big[ Y^{[s]}\,\rho^{[j]} \big]
= \mathbb{E}[Y\,\rho^{[j]}] - \nu^{[j]}\,\mathbb{E}[(\rho^{[j]})^2].
\]

It remains to identify \(\mathbb{E}[Y\,\rho^{[j]}]\). Using the definition of the \((j)\)-th residual response,
\[
Y = Y^{[j]} + \sum_{m=1}^{j-1} \nu^{[m]}\rho^{[m]},
\]
and again using score orthogonality \(\mathbb{E}[\rho^{[m]}\rho^{[j]}]=0\) for \(m<j\), we find
\[
\mathbb{E}[Y\,\rho^{[j]}] = \mathbb{E}[Y^{[j]}\,\rho^{[j]}].
\]
Substitute this into the previous display and use the definition of \(\nu^{[j]}\):
\[
\mathbb{E}\big[ Y^{[s]}\,\rho^{[j]} \big]
= \mathbb{E}[Y^{[j]}\,\rho^{[j]}] - 
\frac{\mathbb{E}[Y^{[j]}\rho^{[j]}]}{\mathbb{E}[(\rho^{[j]})^2]}\,\mathbb{E}[(\rho^{[j]})^2]
= 0.
\]
This completes the proof of orthogonality between scores and residualized responses.
\end{proof}

\subsubsection{Equivalence of the problems}\label{equiv_problem}
\begin{proof}
Recall the residualization definitions (for fixed \(l\ge 2\)):
\[
W^{[l]} \;=\; W - \sum_{j=1}^{\,l-1} \rho^{[j]}\,\delta^{[j]},
\qquad
Y^{[l]} \;=\; Y - \sum_{k=1}^{\,l-1} \nu^{[k]}\,\rho^{[k]},
\]
where
\[
\rho^{[j]}=\langle W^{[j]},\xi^{[j]}\rangle,\quad
\delta^{[j]}=\frac{\mathbb{E}[\rho^{[j]} W^{[j]}]}{\mathbb{E}[(\rho^{[j]})^2]},
\quad
\nu^{[j]}=\frac{\mathbb{E}[\rho^{[j]} Y^{[j]}]}{\mathbb{E}[(\rho^{[j]})^2]}.
\]
Since \(W,Y\) are centered, 
so are $W^{l}, Y^{l}$.
Thus covariances are expectations of products.
 Fix any \(w\in\mathbb{H}\).  Using linearity,
\begin{align*}
 \operatorname{Cov}\!\big(\langle W^{[l]},w\rangle,\;Y^{[l]}\big)
&= \mathbb{E}\Big[\big(\langle W,w\rangle - \sum_{j=1}^{l-1} \rho^{[j]}\langle\delta^{[j]},w\rangle\big)
\big(Y - \sum_{k=1}^{l-1} \nu^{[k]}\rho^{[k]}\big)\Big] \\
\\&= \mathbb{E}[\langle W,w\rangle Y] 
- \sum_{k=1}^{l-1} \nu^{[k]}\,\mathbb{E}[\langle W,w\rangle \rho^{[k]}]
- \underbrace{
\sum_{j=1}^{l-1} \langle\delta^{[j]},w\rangle\,\mathbb{E}[\rho^{[j]} Y]}_{:=S_1} \\
&\qquad
\underbrace{
+ 
\sum_{j=1}^{l-1}\sum_{k=1}^{l-1} \langle\delta^{[j]},w\rangle\,\nu^{[k]}\,\mathbb{E}[\rho^{[j]}\rho^{[k]}]
}_{:=S_2}.
\end{align*}
We claim that $S_1  = S_2$.
For $S_1$, first note that, since
\(Y = Y^{[l]} + \sum_{k=1}^{l-1}\nu^{[k]}\rho^{[k]}\), for each \(j\le l-1\) we have
\[
\mathbb{E}[\rho^{[j]} Y] = \mathbb{E}[\rho^{[j]} Y^{[l]}] + \sum_{k=1}^{l-1} \nu^{[k]}\mathbb{E}[\rho^{[j]}\rho^{[k]}]=
 \nu^{[j]}\,\mathbb{E}[(\rho^{[j]})^2],
\]
where the last equality uses 
\(\mathbb{E}[\rho^{[j]} Y^{[l]}]=0\) (Appendix \ref{section:proof:orthogonal_score_residual_pop}) and
\(\mathbb{E}[\rho^{[j]}\rho^{[k]}]=0\) for \(j\ne k\) (Appendix \ref{section:proof:orthogonal_score_pop}).
Therefore we have
\[
S_1 = \sum_{j=1}^{l-1} \langle\delta^{[j]},w\rangle \nu^{[j]}\,\mathbb{E}[(\rho^{[j]})^2].
\]
For $S_2$, 
since  \(\mathbb{E}[\rho^{[j]}\rho^{[k]}]=0\) for \(j\ne k\) (Appendix \ref{section:proof:orthogonal_score_pop}), we have:
\begin{equation*}
S_2
=
 \sum_{j=1}^{l-1}\sum_{k=1}^{l-1} \langle\delta^{[j]},w\rangle\,\nu^{[k]}\,\mathbb{E}[\rho^{[j]}\rho^{[k]}]
= \sum_{j=1}^{l-1} \langle\delta^{[j]},w\rangle\,\nu^{[j]}\,\mathbb{E}[(\rho^{[j]})^2].   
\end{equation*}
Since $S_1 = S_2$, 
and
\(\mathbb{E}[\langle W,w\rangle Y]=\operatorname{Cov}(\langle W,w\rangle,Y)\), we have
\[
\operatorname{Cov}\big(\langle W^{[l]},w\rangle,\;Y^{[l]}\big)
= \operatorname{Cov}(\langle W,w\rangle,Y)
- \sum_{k=1}^{l-1} \nu^{[k]}\,\operatorname{Cov}(\langle W,w\rangle,\rho^{[k]}).
\]
Therefore we are solving
\begin{align*}
\max_{\|w\|_{\mathbb{H}}=1}~&
\Big\{ 
\operatorname{Cov}(\langle W,w\rangle,Y)
- \sum_{k=1}^{l-1} \nu^{[k]}\,\operatorname{Cov}(\langle W,w\rangle,\rho^{[k]})
\Big\}^2
\\
s.t.~&
\nu^{[k]}
=
\frac{1}{\mathbb{E}[ (\rho^{[k]})^2 ]}
\mathbb{E}[Y^{[l]} \rho^{[k]}],~k=1, \ldots, l-1
\\~&
\rho^{[k]} = \langle \xi^{[k]}, W^{[k]} \rangle,~k=1, \ldots, l-1 
\\~&
\delta^{[k]}
=
\frac{1}{\mathbb{E}[ (\rho^{[k]})^2 ]}  \mathbb{E}[W^{[k]} \rho^{[k]}],~k=1, \ldots, l-1
\\~&
 W^{[k+1]} = W^{[k]} - \rho^{[k]} \, \delta^{[k]},~k=1, \ldots, l-1
\\~&
Y^{[k]} =Y^{[k-1]} - \nu^{[k-1]} \, \rho^{[k-1]},~k=1, \ldots, l-1
\end{align*}
Since the constraints imply 
that for every \(s\ge 1\) and every \(j\le s-1\), we have $
\mathbb{E}[\rho^{[j]}\rho^{[s]}]=\langle w, \Sigma_W \, \xi^{[k]} \rangle_{\mathbb{H}} = \operatorname{Cov}(\langle W,w\rangle,\rho^{[k]}) =0
$ (Appendix \ref{section:proof:orthogonal_score_pop}),
the above problem is equivalent to 
\begin{align*}
\max_{\|w\|_{\mathbb{H}}=1}~&
\Big\{ 
\operatorname{Cov}(\langle W,w\rangle,Y)
\Big\}^2
\\
s.t.~&
\nu^{[k]}
=
\frac{1}{\mathbb{E}[ (\rho^{[k]})^2 ]}
\mathbb{E}[Y^{[l]} \rho^{[k]}],~k=1, \ldots, l-1
\\~&
\rho^{[k]} = \langle \xi^{[k]}, W^{[k]} \rangle,~k=1, \ldots, l-1 
\\~&
\delta^{[k]}
=
\frac{1}{\mathbb{E}[ (\rho^{[k]})^2 ]}  \mathbb{E}[W^{[k]} \rho^{[k]}],~k=1, \ldots, l-1
\\~&
 W^{[k+1]} = W^{[k]} - \rho^{[k]} \, \delta^{[k]},~k=1, \ldots, l-1
\\~&
Y^{[k]} =Y^{[k-1]} - \nu^{[k-1]} \, \rho^{[k-1]},~k=1, \ldots, l-1
\\~& \langle w, \Sigma_W \, \xi^{[k]} \rangle_{\mathbb{H}}=0,~k=1, \ldots, l-1.
\end{align*}
Removing constraints that does not affect the objective function, the above optimization problem is equivalent to
\begin{align*}
\max_{\|w\|_{\mathbb{H}}=1}~&
\Big\{ 
\operatorname{Cov}(\langle W,w\rangle,Y)
\Big\}^2
\\
s.t.~&
 \langle w, \Sigma_W \, \xi^{[k]} \rangle_{\mathbb{H}}=0,~k=1, \ldots, l-1.
\end{align*}
This completes
the proof of Proposition \ref{prop:residualization_equiv_eigen}.
\end{proof}

\subsection{Proof of Theorem \ref{thm:hybrid_pls_convergence}}\label{section:proof:thm:hybrid_pls_convergence}
\begin{proof}
Iterating the residualization $Y^{[l]} = Y^{[l-1]} - \nu^{[l-1]} \rho^{[l-1]}$ for $l=1,\dots,L$, we obtain
\begin{equation}\label{eq:Y_sum_residuals}
Y = \sum_{j=0}^{L-1} \nu^{[j]} \rho^{[j]} + Y^{[L]} = Y_{\mathrm{PLS},L-1} + Y^{[L]}.
\end{equation}
Consider the $\mathcal{L}^2$ norm of the remaining residual $Y^{[L]}$:
\[
\mathbb{E}[(Y^{[L]})^2] = \mathbb{E}[(Y^{[L-1]} - \nu^{[L-1]} \rho^{[L-1]})^2] = \mathbb{E}[(Y^{[L-1]})^2] - (\nu^{[L-1]})^2 \mathbb{E}[(\rho^{[L-1]})^2].
\]
Thus, the residual variance is nonnegative and monotonically decreasing with $L$:
\[
0 \le \mathbb{E}[(Y^{[L]})^2] \le \mathbb{E}[(Y^{[L-1]})^2] \le \cdots \le \mathbb{E}[Y^2] < \infty.
\]
Therefore, the series $\sum_{l=1}^{\infty} (\nu^{[l]})^2 \mathbb{E}[(\rho^{[l]})^2]$ converges and the residual variance $\mathbb{E}[(Y^{[L]})^2] \to 0$ as $L \to \infty$.
From \eqref{eq:Y_sum_residuals}, the mean squared error of the PLS approximation is
\[
\mathbb{E}\bigl[ \| Y_{\mathrm{PLS},L-1} - Y \|^2 \bigr] = \mathbb{E}[(Y^{[L]})^2] \to 0 \quad \text{as } L \to \infty.
\]
This completes the proof of Theorem \ref{thm:hybrid_pls_convergence}.
\end{proof}

\section{Proof of Section \ref{section:sub:iterative}}\label{section:proof:section:sub:iterative}

	This section provides the proofs for the core problem formulations and computational schemes of the iterative steps for estimating the PLS direction and performing residualization.
	Appendix   \ref{section:proof:proposition:eigen_regul} derive the eigenproblem formulations for  regularized PLS direction estimation.
	Appendices \ref{section:proof:proposition:linear_regul} and \ref{section:proof:proposition:closed_form_orthgonalization} provide proofs for the simple computation schemes for regularized PLS direction estimation and residualization, respectively.

		\subsection{Proof of Proposition \ref{proposition:eigen_regul} }\label{section:proof:proposition:eigen_regul}
We begin by establishing the result for the unregularized case, which serves as the foundation for the general problem. The objective is to estimate the hybrid PLS direction $\xi \in \widetilde{\mathbb{H}}$ by maximizing the squared empirical covariance with the response. This approach fully leverages the continuous nature of the functional components and the hybrid structure of the data.The PLS direction is estimated by the unit-norm vector $\xi \in \widetilde{\mathbb{H}}$ that maximizes the squared empirical covariance between the PLS scores $\langle \widetilde{W}_i, \xi \rangle_{\mathbb{H}}$ and the responses $Y_i$:
\begin{equation}\label{def:squared_empirical_cov}\widehat{\textnormal{Cov}}( \langle\widetilde{W},\xi \rangle_{\mathbb{H}}, Y ):=\frac{1}{n}\sum_{i=1}^ny_i\langle \widetilde{W}_i, \xi \rangle_{\mathbb{H}}.\end{equation}The estimator is defined as:\begin{equation}\label{def:maximizer_squared_empirical_cov}\hat{\xi}:= \arg \max_{\xi \in \widetilde{\mathbb{H}}}~\widehat{\textnormal{Cov}}^2\bigl(\langle \widetilde{W}, \xi \rangle_{\mathbb{H}}, Y\bigr) \quad \text{s.t.} \quad \| \xi  \|_\mathbb{H} = 1.\end{equation}The candidate direction $\xi$ admits the finite-basis expansion:\begin{equation*}\xi =\bigl(\xi_1(t), \ldots, \xi_K(t), \boldsymbol{\zeta}\bigr)=\biggl(\sum_{m=1}^M \gamma_{1m} , b_m(t),\ldots,\sum_{m=1}^M \gamma_{Km} , b_m(t), \boldsymbol{\zeta}\biggr).\end{equation*}Solving for the optimal coefficients is equivalent to the maximization problem in \eqref{def:maximizer_squared_empirical_cov}. The following proposition reformulates this procedure as a generalized Rayleigh quotient.
	\begin{theorem}\label{proposition:eigen_noregul}
		Let
		$
		(\mathbb{B}, \mathbb{B}^{\prime \prime}, \Theta, \mathbf{y})
		$
		denote the observed data defined in \eqref{def:problem_instance}.
		At the $l$-th iteration of the PLS algorithm, 
		the coefficients of the squared covariance  maximizer defined in \eqref{def:maximizer_squared_empirical_cov}, 
		is obtained as
		\begin{equation}	\label{eq: Unregularized generalized rayleigh quotient equation}
			\left(
			\hat{\gamma}_{11}, \ldots, \hat{\gamma}_{1M}
			, \ldots,
			\hat{\gamma}_{K1}, \ldots, \hat{\gamma}_{KM}
			, \hat{\boldsymbol{\zeta}}^\top
			\right)^\top
			=
			\arg \max_{\boldsymbol{\xi} \in \mathbb{R}^{MK+p}}   \boldsymbol{\xi}^\top \mathbf{V} \boldsymbol{\xi}
			\quad \text{subject to} \quad \boldsymbol{\xi}^\top \mathbb{B} \boldsymbol{\xi} = 1.
		\end{equation}
	\end{theorem}

\begin{proof}
		Let $\xi = (\xi_1, \ldots, \xi_K, \boldsymbol{\zeta}) \in \mathbb{H}$. Assume that each functional component $\xi_j \in \mathbb{L}^2([0,1])$ lies in the span of the basis functions ${b_1(t), \ldots, b_M(t)}$ and can be written as
		\begin{equation*}
			\xi_j(t) := 
			\sum_{m=1}^M d_{jm} \, b_m(t), \quad j = 1, \ldots, K,
		\end{equation*}
		with coefficient vectors $\boldsymbol{\gamma}_j := (d_{j1}, \ldots, d_{jM})^\top \in \mathbb{R}^M$.

        \paragraph{Objective function}~
	The empirical covariance is
		computed as follows:
		\begin{align*}
			\widehat{\textnormal{Cov}} (
			\langle \widetilde{W}, \xi  \rangle_{\mathbb{H}}, Y 
			) 
			&=  %1
			\frac{1}{n} 
			\sum 
			\limits_{i=1}^n 
			y_i
			\langle \widetilde{W}_i, \xi \rangle_\mathbb{H} \\
			&= %2
			\frac{1}{n} 
			\sum \limits_{i=1}^n
			y_i
			\biggl(
			\sum_{j=1}^K
			\langle \widetilde{X}_{ij}, \xi_j \rangle  
			+
			\mathbf{Z}_i^\top 
			\boldsymbol{\zeta} 
			\biggr)
			%%%%%%%%%%%%%%%%%%%%%%%%%%%%%%%%%%%%%%
			\\&= %3
			\frac{1}{n}
			\sum
			\limits_{i=1}^n 
			y_i
			\biggl( \sum \limits_{k=1}^K \int_0^1 \widetilde{X}_{ij}(t) \, \xi_j(t) \, dt \biggr)
			+
			\frac{1}{n}
			\sum
			\limits_{i=1}^n
			y_i
			( \mathbf{Z}_i^\top \boldsymbol{\zeta} )
			%%%%%%%%%%%%%%%%%%%%%%%%%%%%%%%%%%%%%%%%%%
			\\&=  
			\frac{1}{n} \sum \limits_{i=1}^n 
			y_i
			\biggl\{ 
			\sum \limits_{k=1}^K \int_0^1
			\biggl(
			\sum_{m=1}^M \theta_{ijm} b_{m}(t) 
			\biggr) 
			\biggl(
			\sum_{m'=1}^M d_{jm'} b_{m'}(t) 
			\biggr) 
			dt
			\biggr\}  
			+
			\frac{1}{n}
			\mathbf{y}^\top
			Z \boldsymbol{\zeta} 
			%%%%%%%%%%%%%%%%%%%%%%%%%%%%%%%%%%%%%%%%%%
			\\&=  
			\frac{1}{n} \sum \limits_{i=1}^n 
			y_i
			\biggl\{ 
			\sum \limits_{k=1}^K 
			\sum_{m=1}^M
			\sum_{m'=1}^M  
			\theta_{ijm}
			d_{jm'}
			\int_0^1
			b_{m}(t) 
			b_{m'}(t) 
			dt
			\biggr\}  
			+
			\frac{1}{n}
			\mathbf{y}^\top
			Z \boldsymbol{\zeta} 
			%%%%%%%%%%%%%%%%%%%%%%%%%%%%%%%%%%%%%%%%%%
			\\&=  
			\sum \limits_{k=1}^K 
			\biggl\{
			\frac{1}{n} \sum \limits_{i=1}^n 
			y_i
			(
			\boldsymbol{\theta}_{ij}^\top
			B
			\boldsymbol{\gamma}_j
			)  
			\biggr\}
			+
			\frac{1}{n}
			\mathbf{y}^\top
			Z \boldsymbol{\zeta}
			%%%%%%%%%%%%%%%%%%%%%%%%%%%%%%%%%%%%%%%%%%
			\\&=  
			\frac{1}{n}
			\biggl(
			\sum \limits_{k=1}^K 
			\mathbf{y}^\top
			\Theta_{j}
			B
			\boldsymbol{\gamma}_j
			+
			\mathbf{y}^\top
			Z \boldsymbol{\zeta} 
			\biggr)
		\end{align*}	
		Building on this computation, the squared empirical covariance is
		expressed as the quadratic form involving the matrix $\mathbf{V}$.
		%	The objective function matches that of \eqref{eq: Regularized generalized rayleigh quotient equation}, takin
		\begin{align*}
			\widehat{\textnormal{Cov}}^2 (
			\langle \widetilde{W}, \xi  \rangle_{\mathbb{H}}, Y 
			)  
			&= 
			\frac{1}{n^2}
			\biggl(
			\sum \limits_{k=1}^K 
			\mathbf{y}^\top
			\Theta_{j}
			B
			\boldsymbol{\gamma}_j
			+
			\mathbf{y}^\top
			Z \boldsymbol{\zeta} 
			\biggr)^2
			%%%%%%%%%%%%%%%%%%%%
			\\	&=
			\frac{1}{n^2} 
			(
			\mathbf{y}^\top
			\Theta
			\mathbb{B}  
			\boldsymbol{\xi}
			)^2
			\\
			&	= 
			\frac{1}{n^2} 
			\boldsymbol{\xi}^\top
			(\mathbb{B} \Theta^\top \mathbf{y})(\mathbb{B} \Theta^\top \mathbf{y})^\top 
			\boldsymbol{\xi}
			\\&=
			\boldsymbol{\xi}^\top
			\mathbf{V} 
			\boldsymbol{\xi}.
			\numberthis \label{covariance_squares_as_quadratic_form}
		\end{align*}

        \paragraph{Constraint}~
		The squared norm of $\xi$ in the hybrid Hilbert space $\mathbb{H}$ 
		%can be expressed as the inner product on $(\mathbb{R}^M)^K \times \mathbb{R}^p$ defined in Lemma~\ref{lemma:valid_inner_product_tuple} as follows:
		is computed as follows:
		\begin{align*}
			\langle {\xi},{\xi} \rangle_{\mathbb{H}}
			&=
			\sum_{j=1}^K
			\int_0^1
			\xi_j(t) \, \xi_j(t)
			\, dt
			+
			\boldsymbol{\zeta}^\top \boldsymbol{\zeta}
			%%%%%%%%%%%%%%%%%%%%%%%%%%%%%%%%%%%
			\\&=
			\sum \limits_{j=1}^K \int_0^1
			\biggl(
			\sum_{m=1}^M d_{jm} b_{m}(t) 
			\biggr) 
			\biggl(
			\sum_{m'=1}^M d_{jm'} b_{m'}(t) 
			\biggr) 
			dt
			+
			\boldsymbol{\zeta}^\top \boldsymbol{\zeta}  
			%%%%%%%%%%%%%%%%%%%%%%%%%%%%%%%%%%%
			\\&=
			\sum \limits_{j=1}^K 
			\sum_{m=1}^M
			\sum_{m'=1}^M
			d_{jm}
			d_{jm'}
			\int_0^1
			b_{m}(t) 
			b_{m'}(t) 
			dt
			+
			\boldsymbol{\zeta}^\top \boldsymbol{\zeta}
			%%%%%%%%%%%%%%%%%%%%%%%%%%%%%%%%%%%
			\\&=
			\sum \limits_{j=1}^K 
			\boldsymbol{\gamma}_j^\top
			B
			\boldsymbol{\gamma}_j
			+
			\boldsymbol{\zeta}^\top \boldsymbol{\zeta},
			%%%%%%%%%%%%%
			\\&=
			\boldsymbol{\xi}^\top \mathbb{B} \boldsymbol{\xi}
			\numberthis \label{unit_norm_as_quadform}
		\end{align*}
		Therefore, the covariance maximization problem \eqref{def:maximizer_squared_empirical_cov} is equivalent to the generalized Raleigh quotient \eqref{eq: Unregularized generalized rayleigh quotient equation}. This completes  the proof of Proposition~\ref{proposition:eigen_noregul}.
       \end{proof} 

    Now, for the regularized case, we only need to deal with the new constraint.
		The penalization term 
		can be written in matrix form as:
		\begin{align*}
			\sum \limits_{j=1}^K 
			\lambda_j
			\int_0^1 \bigl\{ \hat{\xi}_j^{\prime\prime}(t) \bigr\}^2 dt
			&=
			\sum \limits_{j=1}^K 
			\lambda_j
			\int_0^1 \bigl\{ \sum_{l=1}^M \theta_{ijl} \, b_l^{\prime \prime}(t) \bigr\}^2 dt
			\\&	 =
			%%%%%%%%%%%%%%%%%%%
			\sum \limits_{j=1}^K 
			\lambda_j
			\sum_{l=1}^M
			\sum_{m=1}^M
			\theta_{ijl} \theta_{ijm}
			\int_0^1
			\, b_l^{\prime \prime}(t) 
			b_m^{\prime \prime}(t)
			dt
			\\&	 =
			%%%%%%%%%%%%%%%%%%%%%%%%%%%%%%%%%
			\sum \limits_{j=1}^K
			\lambda_j
			\boldsymbol{\gamma}_j^\top
			B^{\prime \prime} 
			\boldsymbol{\gamma}_j
			%%%%%%%%%%%%%%%%%%%%%%%%%%
			\\&= 
			\boldsymbol{\xi}^{\top} \Lambda \mathbb{B}^{\prime \prime} \boldsymbol{\xi},
			\numberthis \label{pen_as_matrix}
		\end{align*}
		where    $\mathbb{B}^{\prime\prime}$ is defined in \eqref{def:gram_block}, and $\Lambda$ is defined in \eqref{def:Lambda}. 
		Combining this with \eqref{unit_norm_as_quadform}, 
		we can compute
		\begin{equation*}
			\| \xi  \|_\mathbb{H} + \sum \limits_{j=1}^K 
			\lambda_j
			\int_0^1 \bigl\{ \hat{\xi}_j^{\prime\prime}(t) \bigr\}^2 dt
			= 
			\boldsymbol{\xi}^\top \mathbb{B} \boldsymbol{\xi}
			+ 
			\boldsymbol{\xi}^{\top} \Lambda \mathbb{B}^{\prime \prime} \boldsymbol{\xi}
			=
			\boldsymbol{\xi}^\top 
			(
			\mathbb{B}
			+
			\Lambda \mathbb{B}^{\prime \prime})
			\boldsymbol{\xi}.
		\end{equation*}
		This computation along with \eqref{covariance_squares_as_quadratic_form} implies that the covariance maximization problem  
		\eqref{def:maximizer_squared_empirical_cov_reg} is equivalent to the generalized Raleigh quotient \eqref{eq: Regularized generalized rayleigh quotient equation}.  This completes  the proof of Proposition \ref{proposition:eigen_regul}.

	\subsection{Proof of Proposition  \ref{proposition:linear_regul}}\label{section:proof:proposition:linear_regul}
	\begin{proof}
		The optimization problem 
		\eqref{eq: Regularized generalized rayleigh quotient equation} can be written as
		\begin{align*}
			\max_{\boldsymbol{\gamma}_1, \ldots, \boldsymbol{\gamma}_K \in \mathbb{R}^M,\;
				\boldsymbol{\zeta} \in \mathbb{R}^p}~&
			\frac{1}{n^2}
			\left(
			\sum_{j=1}^K 
			\mathbf{y}^\top
			\Theta_j
			B  
			\boldsymbol{\gamma}_j
			+
			\mathbf{y}^\top
			Z \boldsymbol{\zeta} 
			\right)^2,
			\\
			\text{subject to}\quad&
			\sum_{j=1}^K 
			\boldsymbol{\gamma}_j^\top
			(\mathbf{B} + \lambda_j \mathbf{B}^{\prime\prime})
			\boldsymbol{\gamma}_j
			+
			\boldsymbol{\zeta}^\top \boldsymbol{\zeta} = 1.
		\end{align*}
		
		Let $\mathbf{u}_j := B \Theta_j^\top \mathbf{y} \in \mathbb{R}^M$ and $\mathbf{v} := Z^\top \mathbf{y} \in \mathbb{R}^p$. These are fixed problem data. Then the objective becomes
		\[
		\frac{1}{n^2}
		\left(
		\sum_{j=1}^K 
		\mathbf{u}_j^\top \boldsymbol{\gamma}_j
		+
		\mathbf{v}^\top \boldsymbol{\zeta}
		\right)^2.
		\]
		The objective function is continuous. The constraint defines the boundary of an ellipsoid in a finite-dimensional Euclidean space, the feasible set is compact. By the Weierstrass Extreme Value Theorem (see, e.g., Theorem 2.3.1 in \citealp{bazaraaNonlinearProgrammingTheory2006a}), a global maximizer exists. The constraint function is continuously differentiable and its gradient vanishes only at the origin, which is not feasible. 
		Thus the gradient  is nonzero at all feasible points.
		Hence, the Linear Independence Constraint Qualification (LICQ) holds, and the Karush-Kuhn-Tucker (KKT) conditions are necessary for local optimality (see, e.g., Theorem 5.3.1 in \citealp{bazaraaNonlinearProgrammingTheory2006a}).
		
		Define the Lagrangian:
		\[
		\mathcal{L}(\{\boldsymbol{\gamma}_j\}, \boldsymbol{\zeta}, \mu) 
		:= 
		\left( \sum_{j=1}^K \mathbf{u}_j^\top  \boldsymbol{\gamma}_j + \mathbf{v}^\top \boldsymbol{\zeta} \right)^2
		- 
		\mu 
		\left[
		\sum_{j=1}^K \boldsymbol{\gamma}_j^\top (\mathbf{B} + \lambda_j \mathbf{B}^{\prime\prime}) \boldsymbol{\gamma}_j + \boldsymbol{\zeta}^\top \boldsymbol{\zeta} - 1
		\right].
		\]
		
		Let $s := \dfrac{2}{n^2} \left( \sum_{j=1}^K \mathbf{u}_j^\top \boldsymbol{\gamma}_j + \mathbf{v}^\top \boldsymbol{\zeta} \right)$. The KKT conditions require that:
		\begin{align}
			\nabla_{\boldsymbol{\gamma}_j} \mathcal{L} &= 
			s \mathbf{u}_j - 2\mu (\mathbf{B} + \lambda_j \mathbf{B}^{\prime\prime}) \boldsymbol{\gamma}_j = 0, \quad \text{for } j = 1, \ldots, K, \label{eq:lag_grad_gamma} \\
			\nabla_{\boldsymbol{\zeta}} \mathcal{L} &= 
			s \mathbf{v} - 2\mu \boldsymbol{\zeta} = 0. \label{eq:lag_grad_zeta}
		\end{align}
		
		From \eqref{eq:lag_grad_gamma} and \eqref{eq:lag_grad_zeta}, we have
		\[
		(\mathbf{B} + \lambda_j \mathbf{B}^{\prime\prime}) \boldsymbol{\gamma}_j = \frac{s}{2\mu} \mathbf{u}_j, \quad
		\boldsymbol{\zeta} = \frac{s}{2\mu} \mathbf{v}, \quad \text{for } j = 1, \ldots, K.
		\]
		This implies that for any local maximizer, there exists $c \ne 0$ such that 
		\[
		\boldsymbol{\gamma}_j = c (\mathbf{B} + \lambda_j \mathbf{B}^{\prime\prime})^{-1} \mathbf{u}_j, \quad
		\boldsymbol{\zeta} = c \mathbf{v}, \quad \text{for } j = 1, \ldots, K.
		\]
		Since we assumed in Section~\ref{sec: finite basis approximation} that the functions $\{b_m(t)\}$ and their second derivatives are linearly independent, both $B$ and $B^{\prime\prime}$ are positive definite (see, for example, Theorem 273 of \citealp{gockenbachFiniteDimensionalLinearAlgebra2010}). As positive definiteness is preserved under conic combinations, $B + \lambda_j B^{\prime\prime}$ is also positive definite. 
		
		A local maximizer must also be primal feasible. Substituting the conditions above in to the primal constraint:
		\begin{align*}
			\sum_{j=1}^K 
			\left( c (\mathbf{B} + \lambda_j \mathbf{B}^{\prime\prime})^{-1} \mathbf{u}_j \right)^\top 
			(\mathbf{B} + \lambda_j \mathbf{B}^{\prime\prime}) 
			\left( c (\mathbf{B} + \lambda_j \mathbf{B}^{\prime\prime})^{-1} \mathbf{u}_j \right)
			+ c^2 \mathbf{v}^\top \mathbf{v} &= 1 \\
			\Rightarrow 
			c^2 \left( \sum_{j=1}^K \mathbf{u}_j^\top (\mathbf{B} + \lambda_j \mathbf{B}^{\prime\prime})^{-1} \mathbf{u}_j + \mathbf{v}^\top \mathbf{v} \right) &= 1.
		\end{align*}
		
		Solving for $c^2$ gives:
		\[
		c^2 = \left( \sum_{j=1}^K \mathbf{u}_j^\top (\mathbf{B} + \lambda_j \mathbf{B}^{\prime\prime})^{-1} \mathbf{u}_j + \mathbf{v}^\top \mathbf{v} \right)^{-1}.
		\]
		
		Hence, the unique maximizer (up to sign) is:
		\[
		\boldsymbol{\gamma}_j^\ast = \frac{1}{\sqrt{q}} (\mathbf{B} + \lambda_j \mathbf{B}^{\prime\prime})^{-1} \mathbf{u}_j, \quad
		\boldsymbol{\zeta}^\ast = \frac{1}{\sqrt{q}} \mathbf{v},
		\]
		where
		\[
		q := \sum_{j=1}^K \mathbf{u}_j^\top (\mathbf{B} + \lambda_j \mathbf{B}^{\prime\prime})^{-1} \mathbf{u}_j + \mathbf{v}^\top \mathbf{v}.
		\]
		
		This completes the proof of Proposition \ref{section:proof:proposition:linear_regul}.

	\end{proof}

	\subsection{Proof of Lemma \ref{proposition:closed_form_orthgonalization} }\label{section:proof:proposition:closed_form_orthgonalization}
	For notational simplicity, we omit the iteration superscript $[l]$.
	Recall from  \eqref{def:plsscore} that
	$\hat{\rho}_i
	=
	\langle \widetilde{W}_i, \, \hat{\xi} \rangle_{\mathbb{H}}$.
	Using the basis expansion coefficient notation from 
    Section
    \ref{sec: finite basis approximation}, the full vector of PLS scores $\hat{\boldsymbol{\rho}}  = (\hat{\rho}_1, \ldots, \hat{\rho}_n)^\top$ 
	is computed through the following matrix multiplication:
	\begin{equation}\label{score_inner_product_all_n}
		\hat{\boldsymbol{\rho}}^\top = \sum_{k=1}^K ( \hat{\boldsymbol{\gamma}}_j )^\top \mathbf{B} \, \Theta_j^\top + \boldsymbol{\zeta}^\top \mathbf{Z}^\top.
	\end{equation}
	The least square criterion   can be decomposed as follows:
	\begin{align*}
		\operatorname{SSE}( \delta) &= \sum_{i=1}^n \langle \widetilde{W}_i^{[l]} - \hat{\rho}_i^{[l]}\, \delta, \, \widetilde{W}_i^{[l]} - \hat{\rho}_i^{[l]}\, \delta \rangle_{\mathbb{H}} \\
		&= 
		\underbrace{
			\sum_{i=1}^n \langle \widetilde{W}_i^{[l]}, \, \widetilde{W}_i^{[l]} \rangle_{\mathbb{H}}}_{A} 
		\underbrace{
			-2 \sum_{i=1}^n \hat{\rho}_i^{[l]} \langle \delta, \, \widetilde{W}_i^{[l]} \rangle_{\mathbb{H}} 
		}_{B_1(\boldsymbol{\delta})}
		\underbrace{
			+ 
			(\hat{\boldsymbol{\rho}}^\top \hat{\boldsymbol{\rho}})
			\langle \delta, \, \delta \rangle_{\mathbb{H}}
		}_{B_2(\boldsymbol{\delta})}
	\end{align*}
 
	Here, part $A$ does not contain $\delta$, so it does not contribute to the minimization problem.
	
	Let $\boldsymbol{\delta}:= (\pi_1^\top, \ldots, \pi_K^\top, \boldsymbol{\chi}^\top )^\top$ be the $MK+p$-dimensional concatenated vector of basis coefficients and the scalar part of $\boldsymbol{\delta}$. 
	%Abusing notation, we treat $\operatorname{PENSSE}$ as a function of $\boldsymbol{\delta}$.
	Abusing notation, we treat $\operatorname{SSE}$ as a function of $\boldsymbol{\delta}$.
	Next, We will now demonstrate that the combined term $
	%B_1(\boldsymbol{\delta}) + B_2(\boldsymbol{\delta}) + B_3(\boldsymbol{\delta})
	B_1(\boldsymbol{\delta}) + B_2(\boldsymbol{\delta})
	$ is a quadratic function of $\boldsymbol{\delta}$, by expanding its component functions:
	\begin{align*}
		B_1(\boldsymbol{\delta}) &= 
		- 2\sum_{i=1}^n \hat{\rho}_i  \langle \delta, \, \widetilde{W}_i  \rangle_{\mathbb{H}} 
		=
		- 2\hat{\boldsymbol{\rho}}^\top 
		\biggl(
		\sum_{j=1}^K \Theta_j \mathbf{B} \boldsymbol{\pi}_j
		+\mathbf{Z} \boldsymbol{\chi}
		\biggr),
		%%%%%%%%%%%%%%%%%%%%%%%
		\\ 
		B_2(\boldsymbol{\delta})
		&=
		(\hat{\boldsymbol{\rho}}^\top \hat{\boldsymbol{\rho}})
		\langle \delta, \, \delta \rangle_{\mathbb{H}}
		=
		(\hat{\boldsymbol{\rho}}^\top \hat{\boldsymbol{\rho}}) 
		\biggl(
		\sum_{j=1}^K \pi_j^\top \mathbf{B} \pi_j
		+\boldsymbol{\chi}^\top \boldsymbol{\chi}
		\biggr),
		\\
		%	B_3(\boldsymbol{\delta}) &= \tau \sum_{j=1}^K\operatorname{PEN}(\delta_j) = \tau \boldsymbol{\delta} B^{''} \boldsymbol{\delta},
	\end{align*}
	%where the expansion of $B_3(\boldsymbol{\delta})$ leverages the computation in \eqref{pen_as_matrix}.
	Now we compute the gradients:
	\begin{align*}
		\nabla_{\boldsymbol{\pi}_j}
		B_1(\boldsymbol{\delta}) 
		&= -2
		\mathbf{B}
		\Theta_j^\top
		\hat{\boldsymbol{\rho}}  ,~
		j = 1, \ldots, K,
		\quad
		\nabla_{\boldsymbol{\chi}}
		B_1(\boldsymbol{\delta}) 	  
		= -2\mathbf{Z}^\top  \hat{\boldsymbol{\rho}},
		\\
		\nabla_{\boldsymbol{\pi}_j}
		B_2(\boldsymbol{\delta}) 
		&= 
		2
		(\hat{\boldsymbol{\rho}}^\top \hat{\boldsymbol{\rho}}) 
		\mathbf{B} 	\boldsymbol{\pi}_j
		,~
		j = 1, \ldots, K,
		\quad
		\nabla_{\boldsymbol{\chi}}
		B_2(\boldsymbol{\delta}) 
		=
		2(\hat{\boldsymbol{\rho}}^\top \hat{\boldsymbol{\rho}})\boldsymbol{\chi} ,
		%%%%%%%%%%%%%%%%%%%%%%%%%%%%%%%%
		%\\ \nabla_{\boldsymbol{\pi}_j} B_3(\boldsymbol{\delta})  &=  2 \tau B^{\prime \prime }	\boldsymbol{\pi}_j ,~ j = 1, \ldots, K, \quad \nabla_{\boldsymbol{\chi}} B_3(\boldsymbol{\delta})  =0.
	\end{align*}
	The Hessian of 
	$\operatorname{SSE}( \boldsymbol{\delta})$ is then given by 
	$
	2 	(\hat{\boldsymbol{\rho}}^\top \hat{\boldsymbol{\rho}})  \mathbf{B}.
	$
 
	%Since we assumed in Section \ref{sec: finite basis approximation} that both $B'$ and $\mathbb{B}''$ are positive definite, $\operatorname{PENSSE}( \delta)$ is convex. Consequently, the gradient vanishes at its unique minimizer:
	Since   $\mathbf{B}$   is positive definite, $\operatorname{SSE}( \delta)$ is convex. Consequently, the gradient vanishes at its unique minimizer:
	\begin{align*}
		\nabla_{\boldsymbol{\pi}_j}	\operatorname{SSE}( \hat{\boldsymbol{\delta}} )
		&	=
		-2
		\mathbf{B}
		\Theta_j^\top
		\hat{\boldsymbol{\rho}}
		+
		2	(\hat{\boldsymbol{\rho}}^\top \hat{\boldsymbol{\rho}}) 
		\mathbf{B} 	\hat{\boldsymbol{\pi}}_j
		=
		0,~
		j = 1, \ldots, K,
		\\
		\nabla_{\boldsymbol{\chi} }	\operatorname{SSE}( \hat{\boldsymbol{\delta}} )
		&=
		-2\mathbf{Z}^\top  \hat{\boldsymbol{\rho}}
		+
		2 (\hat{\boldsymbol{\rho}}^\top \hat{\boldsymbol{\rho}}) \hat{\boldsymbol{\chi} }
		=
		0,
	\end{align*}
	 
	providing the following closed-form solution:
 
	\begin{equation*}
		\hat{\boldsymbol{\pi}}_j\\
		=
		\bigl\{
		(\hat{\boldsymbol{\rho}}^\top \hat{\boldsymbol{\rho}})  \mathbf{B}
		\bigr\}^{-1}
		\mathbf{B}
		\Theta_j^\top
		\hat{\boldsymbol{\rho}}
		=
		\frac{1}{\hat{\boldsymbol{\rho}}^\top \hat{\boldsymbol{\rho}}} \Theta_j^\top
		\hat{\boldsymbol{\rho}}
		~
		j = 1, \ldots, K,
		%%%%%%%%%%%%%%%%%%%%%%%%
		\quad
		\hat{\boldsymbol{\chi} }
		=
		\frac{1}{(\hat{\boldsymbol{\rho}}^\top \hat{\boldsymbol{\rho}})}
		\mathbf{Z}^\top  \hat{\boldsymbol{\rho}}.
	\end{equation*}
 
	Expanding with respect to these coefficients, we obtain
	\begin{equation*}
		\delta^{[l]}
		= 
		\frac{1}{\| 	\hat{\boldsymbol{\rho}}^{[l]}\|_2^2}
		\sum_{i=1}^n 
		\widehat{\rho}_i^{[l]}
		\widetilde{W}_i^{[l]}.
	\end{equation*}
	On the other hand since $\mathbf{y}^{[l]}$ and  
	$\hat{\boldsymbol{\rho}}^{[l]}$ are zero-mean,
	it is well known that the least square estimate of linear regression coefficient is computed as $\hat{\nu}^{[l]} =  
	\frac{
		\mathbf{y}^{[l] \top }  \hat{\boldsymbol{\rho}}^{[l]}
	}{
		\| \hat{\boldsymbol{\rho}}^{[l]} \|_2^2}$.
	This completes the proof of Lemma \ref{proposition:closed_form_orthgonalization}.
	
	\section{Proof of Lemma \ref{lemma:recursive}}\label{section:proof:lemma:recursive}
	\begin{proof}We prove the lemma using mathematical induction.
		
		\medskip
		\noindent
		\textit{Base case ($l=1$).}~
		For $l=1$, the lemma states $\widehat{\rho}_i^{[1]} = \langle W_i, \widehat{\iota}^{[1]} \rangle_{\mathbb{H}}$. Since  $\widehat{\iota}^{[1]} := \widehat{\xi}^{[1]}$ and $W^{[1]}_i = W_i$, the base case holds.
		
		\medskip
		\noindent
		\textit{Inductive steps.}~		
		Assume   $\widehat{\rho}_i^{[u]} = \langle W_i, \widehat{\iota}^{[u]} \rangle_{\mathbb{H}}$ for $u =1 , \ldots, l$. We want to show that $\widehat{\rho}_i^{[l+1]} = \langle W_i, \widehat{\iota}^{[l+1]} \rangle_{\mathbb{H}}$.
		Recall from  \eqref{def:plsscore} and Lemma \ref{proposition:closed_form_orthgonalization} that we have
		\begin{equation}
			\widehat{\rho}_i^{[l+1]} = \langle \widetilde{W}_i^{[l+1]}, \widehat{\xi}^{[l+1]} \rangle_{\mathbb{H}},
			\quad
			\widetilde{W}_i^{[l+1]} = \widetilde{W}_i^{[l]} - \widehat{\rho}_i^{[l]} \widehat{\delta}^{[l]}
		\end{equation}
		Combining these two, we have
		\begin{equation}\label{rholp1}
			\widehat{\rho}_i^{[l+1]} = \langle \widetilde{W}_i^{[l]} - \widehat{\rho}_i^{[l]} \widehat{\delta}^{[l]}, \widehat{\xi}^{[l+1]} \rangle_{\mathbb{H}} 
			=
			\langle \widetilde{W}_i^{[l]}, \widehat{\xi}^{[l+1]} \rangle_{\mathbb{H}} - \widehat{\rho}_i^{[l]} \langle \widehat{\delta}^{[l]}, \widehat{\xi}^{[l+1]} \rangle_{\mathbb{H}}.
		\end{equation}
		Next,
		we can express $\widetilde{W}_i^{[l]}$ in terms of the original predictor $W_i$ by recursively applying  Lemma \ref{proposition:closed_form_orthgonalization}:
		$$ \widetilde{W}_i^{[l]} = \widetilde{W}_i^{[l-1]} - \widehat{\rho}_i^{[l-1]} \widehat{\delta}^{[l-1]} = \ldots = W_i - \sum_{u=1}^{l-1} \widehat{\rho}_i^{[u]} \widehat{\delta}^{[u]} $$
		Substituting this into our equation \eqref{rholp1}, we have:
		\begin{align*}
			\widehat{\rho}_i^{[l+1]} 
			&= \left\langle W_i - \sum_{u=1}^{l-1} \widehat{\rho}_i^{[u]} \widehat{\delta}^{[u]}, \widehat{\xi}^{[l+1]} \right\rangle_{\mathbb{H}} - \widehat{\rho}_i^{[l]} \langle \widehat{\delta}^{[l]}, \widehat{\xi}^{[l+1]} \rangle_{\mathbb{H}} 
			\\&=
			\langle W_i, \widehat{\xi}^{[l+1]} \rangle_{\mathbb{H}} - \sum_{u=1}^{l-1} \widehat{\rho}_i^{[u]} \langle \widehat{\delta}^{[u]}, \widehat{\xi}^{[l+1]} \rangle_{\mathbb{H}} - \widehat{\rho}_i^{[l]} \langle \widehat{\delta}^{[l]}, \widehat{\xi}^{[l+1]} \rangle_{\mathbb{H}}
			\\&=
			\langle W_i, \widehat{\xi}^{[l+1]} \rangle_{\mathbb{H}} - \sum_{u=1}^{l} \widehat{\rho}_i^{[u]} \langle \widehat{\delta}^{[u]}, \widehat{\xi}^{[l+1]} \rangle_{\mathbb{H}}
			\\& \overset{(i)}{=}
			\langle W_i, \widehat{\xi}^{[l+1]} \rangle_{\mathbb{H}} - \sum_{u=1}^{l} \langle \widehat{\delta}^{[u]}, \widehat{\xi}^{[l+1]} \rangle_{\mathbb{H}} \langle W_i, \widehat{\iota}^{[u]} \rangle_{\mathbb{H}}
			%%%%%%%%%%%%%%%%%%%%%%%%
			\\& = 
			\langle W_i, \widehat{\xi}^{[l+1]} \rangle_{\mathbb{H}}
			-
			\left\langle W_i,   \sum_{u=1}^{l} \langle \widehat{\delta}^{[u]}, \widehat{\xi}^{[l+1]} \rangle_{\mathbb{H}} \widehat{\iota}^{[u]} \right\rangle_{\mathbb{H}} 
			%%%%%%%%%%%%%%%%%%%%%%%%
			\\& = 
			\left\langle W_i, \widehat{\xi}^{[l+1]} - \sum_{u=1}^{l} \langle \widehat{\delta}^{[u]}, \widehat{\xi}^{[l+1]} \rangle_{\mathbb{H}} \widehat{\iota}^{[u]} \right\rangle_{\mathbb{H}} 
			\\&= \widehat{\rho}_i^{[l+1]}.
		\end{align*}
		where step $(i)$ uses the inductive hypothesis $\langle W_i, \widehat{\iota}^{[u]} \rangle_{\mathbb{H}} = \widehat{\rho}_i^{[u]}$ for all $u \leq l$.
		By the principle of mathematical induction, the lemma holds for all $l \ge 1$. This completes the proof of Lemma \ref{lemma:recursive}.
	\end{proof}

	\section{Proof of geometric properties}\label{section:appendix:geom}

	\subsection{Proof of Proposition \ref{proposition: modified orthnormality of PLS components}}
	\label{section:proof:proposition: modified orthnormality of PLS components}
	
	\begin{proof}
		The unit norm condition is trivially met by the constraint 
		$1 =  	\hat{\boldsymbol{\xi}}_l^\top 
		\hspace{-.2em}
		(\mathbb{B} + \Lambda \mathbb{B}^{\prime \prime}) \hat{\boldsymbol{\xi}}_l$
		enforced in \eqref{eq: Regularized generalized rayleigh quotient equation}, because we have:
		\begin{align*}
			\hat{\boldsymbol{\xi}}_l^\top 
			\hspace{-.2em}
			(\mathbb{B} + \Lambda \mathbb{B}^{\prime \prime}) \hat{\boldsymbol{\xi}}_l
			%%%%%%%%%%%%%%%%%%%%%%%%%%%%
			&= \hat{\boldsymbol{\xi}}^\top_l 
			\hspace{-.2em}
			\mathbb{B}
			\hat{ \boldsymbol{\xi} }_l
			+
			\hat{ \boldsymbol{\xi} }_l^\top 
			\hspace{-.2em}
			\Lambda \mathbb{B}^{\prime \prime} \hat{ \boldsymbol{\xi} }_l
			%%%%%%%%%%%%%%%%%%%%%%%%%%%%%%%%
			\\&\overset{(i)}{=}
			\sum_{j=1}^K
			\int_0^1
			\hat{\xi}_l(t) \, \hat{\xi}_l(t)
			\, dt
			+
			\hat{\boldsymbol{\zeta}}^\top  \hat{\boldsymbol{\zeta}}
			+
			\hat{\boldsymbol{\xi}}^\top 
			\hspace{-.2em}
			\Lambda \mathbb{B}^{\prime \prime} \hat{\boldsymbol{\xi}}
			%%%%%%%%%%%%%%%%%%%%%%%%%%%%%%%%
			\\&\overset{(ii)}{=}
			\sum_{j=1}^K
			\int_0^1
			\hat{\xi}_l(t) \, \hat{\xi}_l(t)
			\, dt
			+
			\hat{\boldsymbol{\zeta}}^\top \hat{\boldsymbol{\zeta}}
			+
			\sum \limits_{j=1}^K 
			\lambda_l
			\int_0^1 \bigl\{ \hat{\xi}_l^{\prime\prime}(t) \bigr\}^2 dt
			%%%%%%%%%%%%%%%%%%%%%%%%%%%%%%%%%%
			\\&=
			\langle \widehat{\xi}_{l}, \widehat{\xi}_{l} \rangle_{\mathbb{H}, \Lambda},
			\numberthis \label{rough_norm_matrix_form}
		\end{align*}
		where step $(i)$ uses \eqref{unit_norm_as_quadform}
		and
		step $(ii)$ uses \eqref{pen_as_matrix}.
		
		Now to switch our gears toward the the orthogonality. For  $l_1 > l_2 $, we have
		\begin{align*}
			\langle \widehat{\xi}_{l_1}, \widehat{\xi}_{l_2} \rangle_{\mathbb{H}, \Lambda}
			& \overset{(i)}{=}
			\hat{\boldsymbol{\xi}}_{l_1}^\top 
			\hspace{-.2em}
			(\mathbb{B} + \Lambda \mathbb{B}^{\prime \prime}) \hat{\boldsymbol{\xi}}_{l_2}
			%%%%%%%%%%%%%%%%%%%%%%%%%%%%%%%%%%%%%%
			\\	& \overset{(ii)}{=} \frac{1}{\kappa_{l_1}} (V^{[l_1]} \widehat{\boldsymbol{\xi}}_{l_1})^\top \hat{\boldsymbol{\xi}}_{l_2}
			%%%%%%%%%%%%%%%%%%%%%%%%%%%%%%%%%%%%%%
			\\	& = \frac{1}{\kappa_{l_1}}
			\widehat{\boldsymbol{\xi}}_{l_1}^\top
			V^{[l_1]} \hat{\boldsymbol{\xi}}_{l_2} 
			%%%%%%%%%%%%%%%%%%%%%%%%%%%%%%%%%%%%%%
			\\	& \overset{(iii)}{=} \frac{1}{ n^2 \kappa_{l_1}}
			\widehat{\boldsymbol{\xi}}_{l_1}^\top
			(\mathbb{B} \Theta^{[l_1]\top} \mathbf{y})(\mathbb{B} \Theta^{[l_1]\top} \mathbf{y})^\top
			\hat{\boldsymbol{\xi}}_{l_2}
			%%%%%%%%%%%%%%%%%%%%%%%%%%%%%%%%%%%%%%
			\\	& = \frac{1}{ n^2 \kappa_{l_1}}
			\widehat{\boldsymbol{\xi}}_{l_1}^\top
			(\mathbb{B} \Theta^{[l_1]\top} \mathbf{y})
			\mathbf{y}^\top
			\Theta^{[l_1] } 
			\mathbb{B}
			\hat{\boldsymbol{\xi}}_{l_2}
			%%%%%%%%%%%%%%%%%%%%%%%%%%%%%%%%%%%%%%
			\\	& = c
			\mathbf{y}^\top
			(
			\Theta^{[l_1] }
			\mathbb{B}
			\hat{\boldsymbol{\xi}}_{l_2}
			),
		\end{align*}
		where
		step $(i)$ uses \eqref{rough_norm_matrix_form},
		step $(ii)$ uses the fact that $\hat{\boldsymbol{\xi}}_{l_1}$ is a generalized eigenvector ($\kappa_{l_1}$ denotes the corresponding generalized eigenvalue), as presented in Proposition \ref{proposition:eigen_regul},
		and step $(iii)$ uses the definition $\mathbf{V}$.
		Here, $c$ represents a scalar that condenses all multiplicative terms preceding $\mathbf{y}^\top$.
		
		Orthogonality can be demonstrated by showing that $\Theta^{[l_1] } \mathbb{B} \hat{\boldsymbol{\xi}}_{l_2} = \mathbf{0} \in \mathbb{R}^n$.
        From the notations and 
		the construction in Section \ref{sec: finite basis approximation} \eqref{def:gram_block}, the
		$i$th entry of $\Theta^{[l_1] }
		\mathbb{B}
		\hat{\boldsymbol{\xi}}_{l_2}$ 
		is
		\begin{align*}
			(\Theta^{[l_1] }
			\mathbb{B}
			\hat{\boldsymbol{\xi}}_{l_2})_i
			&=
			(\theta_{i11}^{[l_1]}, \ldots, \theta_{i1M}^{[l_1]}, \ldots, \theta_{iK1}^{[l_1]}, \ldots, \theta_{iKM}^{[l_1]}, Z_{i1}^{[l_1]}, \ldots, Z_{ip}^{[l_1]})
			\mathbb{B}
			\hat{\boldsymbol{\xi}}_{l_2}
			%%%%%%%%%%%%%%%%%%%%%%%%%%%
			\\&=
			\sum_{k=1}^K
			\boldsymbol{\theta}_{ik}^{[l_1]} \mathbb{B}
			\hat{\boldsymbol{\gamma}}_k^{[l_2]}
			+
			\mathbf{Z}_{i}^{[l_1] \top } \hat{\boldsymbol{\zeta}}^{[l_2]}
			%%%%%%%%%%%%%%%%%%%%%%%%%%%
			\\	&=
			\langle
			\widetilde{W}_i^{[l_1]}, 
			\hat{\xi}^{[l_2]}
			\rangle_{\mathbb{H}},
			\numberthis \label{thetaBzetal2i}
		\end{align*}
		where the last equality uses the computation of the $i$th entry in \eqref{score_inner_product_all_n}. 
        
From \eqref{thetaBzetal2i}, the condition $\Theta^{[l_1] } \mathbb{B} \hat{\boldsymbol{\xi}}_{l_2} = \mathbf{0}$ is satisfied if $\langle \widetilde{W}_i^{[l_1]}, \hat{\xi}^{[l_2]} \rangle_{\mathbb{H}} = 0$ for all $i=1, \ldots, n$, which means that residualization step completely eliminates the information captured by the weight direction $\hat{\xi}^{[l_2]}$. This  property is formalized in the following lemma:

\begin{lemma} \label{lemma:annihilation}
    For any iteration indices $l_1, l_2$ such that $l_1 > l_2 \geq 1$, the residualized predictor data $\widetilde{\boldsymbol{W}}^{[l_1]}$ is orthogonal to the previously estimated PLS   direction $\widehat{\xi}^{[l_2]}$ with respect to the Hilbert space inner product. That is,
    \begin{equation*}
        \langle \widetilde{W}_i^{[l_1]}, \widehat{\xi}^{[l_2]} \rangle_{\mathbb{H}} = 0 \quad \text{for}~i=1, \ldots, n.
    \end{equation*}
   
\end{lemma}
Proof of Lemma \ref{lemma:annihilation}
is provided in Appendix \ref{section:proof:lemma:annihilation}.
Applying Lemma \ref{lemma:annihilation}
to \eqref{thetaBzetal2i}
 completes the proof of Proposition 
		\ref{proposition: modified orthnormality of PLS components}.

	\end{proof}

	\subsection{Proof of Proposition \ref{proposition: orthnormality of PLS scores}}\label{section:proof:proposition: orthnormality of PLS scores}
\begin{proof}
For $l_1 < l_2$, we have
\begin{align*}
\widehat{\boldsymbol{\rho}}^{[l_1 ] \top}
\widehat{\boldsymbol{\rho}}^{[l_2]}
			%%%%%%%%%%%%%%%%%%%%%%%%%%%%%%
			&=  
			\sum_{i=1}^n
			\widehat{ \rho }^{[l_1 ]}_i \,
			\widehat{ \rho }^{[l_2]}_i
			%%%%%%%%%%%%%%%%%%%%%%%%%%%%%%
			=  
			\sum_{i=1}^n
			\widehat{ \rho }^{[l_1 ]}_i \,
			\bigl \langle 	\widetilde{W}^{[l_2]}_i  ,
			\hat{\xi}^{[l_2]}
			\bigr \rangle_{\mathbb{H}}
			%%%%%%%%%%%%%%%%%%%%%%%%%%%%%%
=  
			\big \langle
			\sum_{i=1}^n
			\widehat{ \rho }^{[l_1 ]}_i \,
			\widetilde{W}^{[l_2]}_i  ,
			\hat{\xi}^{[l_2]}
			\big \rangle_{\mathbb{H}}.
		\end{align*}
		Therefore,
		to show that $\widehat{\boldsymbol{\rho}}^{[l_1 ] \top}
		\widehat{\boldsymbol{\rho}}^{[l_2]} = 0$,
		it suffices to show that
		$\sum_{i=1}^n
		\widehat{ \rho }^{[l_1 ]}_i \,
		\widetilde{W}^{[l_2]}_i  = 0 \in   \mathbb{H}  $,
		where this   zero element represents an ordered pair of $K$ zero functions and a zero matrix.
		Using Lemma    \ref{proposition:closed_form_orthgonalization}, notations from 
		\eqref{notation:n_predictors} and
		\eqref{notation:n_inner_products},
		and equation \eqref{regression_relation_with_notation_abuse}, we have:
		\begin{align*}
			\widetilde{\boldsymbol{W}}^{[l_1]}  &= 
			\widetilde{\boldsymbol{W}}^{[l_1-1]} - 
			\biggl(
			\frac{1}{\| 	\hat{\boldsymbol{\rho}}^{[l_1-1]}\|_2^2}
			\hat{\boldsymbol{\rho}}^{[l_1-1] }
			\hat{\boldsymbol{\rho}}^{[l_1-1] \top
			} 
			\biggr)
			\widetilde{\boldsymbol{W}}^{[l_1-1]}
			%%%%%%%%%%%%%%%%%%%%%%
			\\&=	\widetilde{\boldsymbol{W}}^{[l_1-1]} - 
			\biggl(
			\frac{1}{\| 	\hat{\boldsymbol{\rho}}^{[l_1-1]}\|_2^2}
			\langle 	
			\widetilde{\boldsymbol{W}}^{[l_1-1]}  ,
			\hat{\xi}^{[l_1-1]}
			\rangle_{\mathbb{H}}
			\hat{\boldsymbol{\rho}}^{[l_1-1] \top
			} 
			\biggr)
			\widetilde{\boldsymbol{W}}^{[l_1-1]}.
		\end{align*}
		Using this relationship and
		with some abuse of notation, we have:
		\begin{align*}
			\sum_{i=1}^n
			\widehat{ \rho }^{[l_1 ]}_i \,
			\widetilde{W}^{[l_2]}_i 
			&=
			\hat{\boldsymbol{\rho}}^{[l_2] \top}
			\widetilde{\boldsymbol{W}}^{[l_1]}
			\\  &=	
			\hat{\boldsymbol{\rho}}^{[l_2] \top}
			\widetilde{\boldsymbol{W}}^{[l_1-1]} - 
			\hat{\boldsymbol{\rho}}^{[l_2] \top}
			\biggl(
			\frac{1}{\| 	\hat{\boldsymbol{\rho}}^{[l_1-1]}\|_2^2}
			\langle 	
			\widetilde{\boldsymbol{W}}^{[l_1-1]}  ,
			\hat{\xi}^{[l_1-1]}
			\rangle_{\mathbb{H}}
			\hat{\boldsymbol{\rho}}^{[l_1-1] \top
			} 
			\biggr)
			\widetilde{\boldsymbol{W}}^{[l_1-1]}
			\\	&=	
			\hat{\boldsymbol{\rho}}^{[l_2] \top}
			\widetilde{\boldsymbol{W}}^{[l_1-1]} - 
			\frac{
				\langle 	 	\hat{\boldsymbol{\rho}}^{[l_2] \top}
				\widetilde{\boldsymbol{W}}^{[l_1-1]}  ,
				\hat{\xi}^{[l_1-1]}
				\rangle_{\mathbb{H}}
			}{\| 	\hat{\boldsymbol{\rho}}^{[l_1-1]}\|_2^2}
			%%%%%%%%%%%%%%%%%%%%%%%%%
			\bigl(
			\hat{\boldsymbol{\rho}}^{[l_1-1] \top
			} 
			\widetilde{\boldsymbol{W}}^{[l_1-1]}
			\bigr)
			\\	&=	h^{[l_1-1]} \bigl( \hat{\boldsymbol{\rho}}^{[l_2] \top}
			\widetilde{\boldsymbol{W}}^{[l_1-1]} \bigr)~(\text{say}).
		\end{align*}
		Here, the function $h^{[l_1-1]}: \tilde{ \mathbb{H} } \mapsto \widetilde{\mathbb{H}}$ maps the zero element of $\widetilde{\mathbb{H}}$ to itself, where the zero element represents an ordered pair of $K$ zero functions and a zero matrix. Using this relationship, we have
		\begin{align*}
			\hat{\boldsymbol{\rho}}^{[l_2] \top}
			\widetilde{\boldsymbol{W}}^{[l_1]}  
			%%%%%%%%%%%%%%%%
			&=	h^{[l_1-1]} \bigl( \hat{\boldsymbol{\rho}}^{[l_2] \top}
			\widetilde{\boldsymbol{W}}^{[l_1-1]} \bigr)
			%%%%%%%%%%%%%%%%
			\\&=
			h^{[l_1-1]} 
			\circ 
			h^{[l_1-2]}
			\bigl( \hat{\boldsymbol{\rho}}^{[l_2] \top}
			\widetilde{\boldsymbol{W}}^{[l_1-2]} \bigr)
			%%%%%%%%%%%%%%%%
			\\&=
			h^{[l_1-1]} 
			\circ 
			h^{[l_1-2]}
			\circ 
			\ldots
			\circ 
			h^{[l_2]}
			\bigl( \hat{\boldsymbol{\rho}}^{[l_2] \top}
			\widetilde{\boldsymbol{W}}^{[l_2]} \bigr).
		\end{align*}
		Therefore,
		to show that $\widehat{\boldsymbol{\rho}}^{[l_1 ] \top}
		\widehat{\boldsymbol{\rho}}^{[l_2]} = 0$,
		it suffices to show $h^{[l_2]}
		\bigl( \hat{\boldsymbol{\rho}}^{[l_2] \top}
		\widetilde{\boldsymbol{W}}^{[l_2]} \bigr) = 0$:
		\begin{align*}
			h^{[l_2]} \bigl( \hat{\boldsymbol{\rho}}^{[l_2] \top}
			\widetilde{\boldsymbol{W}}^{[l_2]} \bigr)
			&=
			\hat{\boldsymbol{\rho}}^{[l_2] \top}
			\widetilde{\boldsymbol{W}}^{[l_2]} - 
			\frac{
				\langle 	 	\hat{\boldsymbol{\rho}}^{[l_2] \top}
				\widetilde{\boldsymbol{W}}^{[l_2]}  ,
				\hat{\xi}^{[l_2]}
				\rangle_{\mathbb{H}}
			}{\| 	\hat{\boldsymbol{\rho}}^{[l_2]}\|_2^2}
			%%%%%%%%%%%%%%%%%%%%%%%%%
			\bigl(
			\hat{\boldsymbol{\rho}}^{[l_2] \top
			} 
			\widetilde{\boldsymbol{W}}^{[l_2]}
			\bigr)
			%%%%%%%%%%%%%%%%%%%%%%%%%%%%
			\\&=
			\hat{\boldsymbol{\rho}}^{[l_2] \top}
			\widetilde{\boldsymbol{W}}^{[l_2]} - 
			\frac{
				\hat{\boldsymbol{\rho}}^{[l_2] \top}
				\langle 	 
				\widetilde{\boldsymbol{W}}^{[l_2]}  ,
				\hat{\xi}^{[l_2]}
				\rangle_{\mathbb{H}}
			}{\| 	\hat{\boldsymbol{\rho}}^{[l_2]}\|_2^2}
			\bigl(
			\hat{\boldsymbol{\rho}}^{[l_2] \top
			} 
			\widetilde{\boldsymbol{W}}^{[l_2]}
			\bigr)
			%%%%%%%%%%%%%%%%%%%%%%%%%%%%
			\\&=
			\hat{\boldsymbol{\rho}}^{[l_2] \top}
			\widetilde{\boldsymbol{W}}^{[l_2]} - 
			\frac{
				\hat{\boldsymbol{\rho}}^{[l_2] \top}
				\hat{\boldsymbol{\rho}}^{[l_2]  } 
			}{\| 	\hat{\boldsymbol{\rho}}^{[l_2]}\|_2^2}
			\bigl(
			\hat{\boldsymbol{\rho}}^{[l_2] \top
			} 
			\widetilde{\boldsymbol{W}}^{[l_2]}
			\bigr)
			%%%%%%%%%%%%%%%%%%%%%%%%%%%%
			\\&=
			\hat{\boldsymbol{\rho}}^{[l_2] \top}
			\widetilde{\boldsymbol{W}}^{[l_2]} 
			- 
			\hat{\boldsymbol{\rho}}^{[l_2] \top
			} 
			\widetilde{\boldsymbol{W}}^{[l_2]}
			\\&=0.
		\end{align*}
		This completes the proof of Proposition \ref{proposition: orthnormality of PLS scores}.
	\end{proof}

\subsection{Proof of Lemma \ref{lemma:annihilation}}\label{section:proof:lemma:annihilation}
\begin{proof}
 Fix  $l_1 > l_2 \geq 1 $.    
 By Lemma  \ref{proposition:closed_form_orthgonalization},
 we have the following
 recursive relationship:
 \begin{equation}\label{algorithm_step:residualization_predictor:recall}
			\widetilde{W}_i^{[l_1]} 
			:= 
			\widetilde{W}_i^{[l_1-1]}  - 
			\frac{\widehat{\rho}_{i\,l_1-1}}{\| 	\hat{\boldsymbol{\rho}}^{[l_1-1]}\|_2^2}
			\sum_{i=1}^n 
			\widehat{\rho}_{i \,l_1-1}
			\widetilde{W}_i^{[l_1-1]}.
		\end{equation} 	
		Let us denote, with some abuse of notation:
		\begin{equation}\label{notation:n_predictors}
			\widetilde{\boldsymbol{W}}^{[l_1]}  := 
			\bigl(
			\widetilde{W}_1^{[l_1-1]}, \ldots, \widetilde{W}_n^{[l_1-1]}
			\bigr)^\top \in \widetilde{\mathbb{H}}^{\otimes n},
		\end{equation}
		and
		\begin{equation}\label{notation:n_inner_products}
			\langle 	\widetilde{\boldsymbol{W}}^{[l_1]}  ,
			\hat{\xi}^{[l_2]}
			\rangle_{\mathbb{H}}
			:= 
			\biggl(
			\langle 
			\widetilde{W}_1^{[l_1-1]}, 
			\hat{\xi}^{[l_2]}
			\rangle_{\mathbb{H}},
			\ldots, 
			\langle 
			\widetilde{W}_n^{[l_1-1]},  
			\hat{\xi}^{[l_2]}
			\rangle_{\mathbb{H}}
			\biggr)^\top 
			\in \mathbb{R}^n.
		\end{equation}
	
		Using these notations and  \eqref{algorithm_step:residualization_predictor:recall}, we have, with some abuse of notation,
		\begin{equation}\label{regression_relation_with_notation_abuse}
			\widetilde{\boldsymbol{W}}^{[l_1]}  = 
			\widetilde{\boldsymbol{W}}^{[l_1-1]} - 
			\biggl(
			\frac{1}{\| 	\hat{\boldsymbol{\rho}}^{[l_1-1]}\|_2^2}
			\hat{\boldsymbol{\rho}}^{[l_1-1] }
			\hat{\boldsymbol{\rho}}^{[l_1-1] \top
			} 
			\biggr)
			\widetilde{\boldsymbol{W}}^{[l_1-1]},
		\end{equation}
		and thus
		\begin{equation*}
			\langle 	\widetilde{\boldsymbol{W}}^{[l_1]}  ,
			\hat{\xi}^{[l_2]}
			\rangle_{\mathbb{H}}
			= 
			\langle 	\widetilde{\boldsymbol{W}}^{[l_1-1]}  ,
			\hat{\xi}^{[l_2]}
			\rangle_{\mathbb{H}} 
			- 
			\biggl(
			\frac{1
			}{
				\| 	\hat{\boldsymbol{\rho}}^{[l_1-1]}\|_2^2
			}
			\hat{\boldsymbol{\rho}}^{[l_1-1]}
			\hat{\boldsymbol{\rho}}^{[l_1-1] \top}
			\biggr)
			\bigl \langle 
			\widetilde{\boldsymbol{W}}^{[l_1-1]}  ,
			\hat{\xi}^{[l_2]}
			\bigr \rangle_{\mathbb{H}}.
		\end{equation*}
		Clearly, the right-hand side represents a linear operator from $\mathbb{R}^n$ to $\mathbb{R}^n$, which we denote as $P^{[l_1-1]}$.
		Thus, 
		we have $\langle 	\widetilde{\boldsymbol{W}}^{[l_1]}  ,
		\hat{\xi}^{[l_2]}
		\rangle_{\mathbb{H}}  = 
		P^{[l_1-1]} \langle 	\widetilde{\boldsymbol{W}}^{[l_1-1]}  ,
		\hat{\xi}^{[l_2]}
		\rangle_{\mathbb{H}}$.
		Repeatedly using this relationship, we have
		\begin{align*}
			\langle 	\widetilde{\boldsymbol{W}}^{[l_1]}  ,
			\hat{\xi}^{[l_2]}
			\rangle_{\mathbb{H}} 
			&=
			P^{[l_1-1]} \langle 	\widetilde{\boldsymbol{W}}^{[l_1-1]}  ,
			\hat{\xi}^{[l_2]}
			\rangle_{\mathbb{H}} 
			\\
			&=
			P^{[l_1-1]} P^{[l_1-2]}   \langle 	\widetilde{\boldsymbol{W}}^{[l_1-2]}  ,
			\hat{\xi}^{[l_2]}
			\rangle_{\mathbb{H}}
			\\
			&= \underbrace{ P^{[l_1-1]} P^{[l_1-2]} \ldots P^{[l_2 + 1]}}_{:=P}  P^{[l_2]} \langle 	\widetilde{\boldsymbol{W}}^{[l_2]}  ,
			\hat{\xi}^{[l_2]}
			\rangle_{\mathbb{H}}
			\\
			&= P \,\biggl\{
			\langle 	\widetilde{\boldsymbol{W}}^{[l_2]}  ,
			\hat{\xi}^{[l_2]}
			\rangle_{\mathbb{H}} 
			- 
			\biggl(
			\frac{1
			}{
				\| 	\hat{\boldsymbol{\rho}}^{[l_2]}\|_2^2
			}
			\hat{\boldsymbol{\rho}}^{[l_2]}
			\hat{\boldsymbol{\rho}}^{[l_2] \top}
			\biggr)
			\bigl \langle 
			\widetilde{\boldsymbol{W}}^{[l_2]}  ,
			\hat{\xi}^{[l_2]}
			\bigr \rangle_{\mathbb{H}}
			\biggr\}
			\\
			&= P \,\biggl\{
			\hat{\boldsymbol{\rho}}^{[l_2]}
			- 
			\biggl(
			\frac{1
			}{
				\| 	\hat{\boldsymbol{\rho}}^{[l_1-1]}\|_2^2
			}
			\hat{\boldsymbol{\rho}}^{[l_2]}
			\hat{\boldsymbol{\rho}}^{[l_2] \top}
			\biggr)
			\hat{\boldsymbol{\rho}}^{[l_2]}
			\biggr\}
			%%%%%%%%%%%%%%%%%%%%%%%%%%%%%%%%%%
			\\&= P \,\biggl\{
			\hat{\boldsymbol{\rho}}^{[l_2]}
			- 
			\biggl(
			\frac{\hat{\boldsymbol{\rho}}^{[l_2] \top} \hat{\boldsymbol{\rho}}^{[l_2]}
			}{
				\| 	\hat{\boldsymbol{\rho}}^{[l_1-1]}\|_2^2
			}
			\biggr)
			\hat{\boldsymbol{\rho}}^{[l_2]}
			\biggr\}
			%%%%%%%%%%%%%%%%%%%%%%%%%%%%%%%%%%
			\\&= P \,\biggl\{
			\hat{\boldsymbol{\rho}}^{[l_2]}
			- 
			\hat{\boldsymbol{\rho}}^{[l_2]}
			\biggr\}
			=
			P \mathbf{0} = \mathbf{0}.
		\end{align*}
 \end{proof}

\end{appendix}

%% ------------------------------------------------------------------
%% REFERENCES
%% ------------------------------------------------------------------
% \bibliography{references} % Uncomment this line and use your .bib file

% For this template example, we use the 'filecontents' environment 
% to generate a bib file on the fly. 
% You should remove this and use your own .bib file.

\end{document}